%% file: causalOT.tex
\documentclass{article}

\usepackage[]{graphicx}
\usepackage[]{color}

\usepackage{enumitem}
\usepackage{lastpage}
\usepackage[margin=1in]{geometry}
\usepackage{float}
\usepackage[dvipsnames]{xcolor}
\definecolor{niceblue}{HTML}{156ABC}
\usepackage{dunipace}

\usepackage{amssymb,amsmath, bm,amsthm}
\usepackage{dsfont}
\usepackage{ifxetex,ifluatex}
\usepackage{fancyvrb}
\allowdisplaybreaks[1]
\usepackage{tabularx}
\usepackage{listings}
\usepackage[ruled,vlined]{algorithm2e}
\usepackage{caption}
\usepackage{subcaption}
\usepackage{tabulary}
\usepackage{setspace}

\usepackage{tikz}
\usetikzlibrary{bayesnet}
\usetikzlibrary{arrows, decorations.markings,shapes.multipart}
\tikzstyle{vecArrow} = [thick, decoration={markings,mark=at position
	1 with {\arrow[semithick]{open triangle 60}}},
double distance=1.4pt, shorten >= 5.5pt,
preaction = {decorate},
postaction = {draw,line width=1.4pt, white,shorten >= 4.5pt}]
\tikzstyle{innerWhite} = [semithick, white,line width=1.4pt, shorten >= 4.5pt]

\SetAlgoLined
\SetAlgoSkip{bigskip}
\usepackage{afterpage}
\usepackage{url}
\usepackage[%
uniquename=false,
uniquelist=false,
maxcitenames=2,
backend=biber,
backref,
natbib=true,
style=authoryear,
isbn=false,
url=false,
doi=false,
eprint=false,
dashed=true
]{biblatex}
\usepackage{xr}

\usepackage[hidelinks,colorlinks=true,citecolor=niceblue,linkcolor=niceblue]{hyperref}

\DefineBibliographyStrings{english}{%
	backrefpage = {\lowercase{s}ee p.},%
	backrefpages = {\lowercase{s}ee pp.}%
}

\DeclareCiteCommand{\citeyear}
{}
{\bibhyperref{\printdate}}
{\multicitedelim}
{}

\DeclareCiteCommand{\citeyearpar}
{}
{\mkbibparens{\bibhyperref{\printdate}}}
{\multicitedelim}
{}

\newcommand{\blind}{0}

\newcommand{\dpl}{\,d\bar{\boldP}^\lambda}
\newcommand{\dpi}{\,d\bar{\pi}}
\newenvironment{bsmallmatrix}
{\bigl[\begin{smallmatrix}}
	{\end{smallmatrix}\bigr]}

\usepackage{sansmathfonts}

\usepackage[default, scale=.95]{lato}
\usepackage[T1]{fontenc}

\begin{document}

	\def\spacingset#1{\renewcommand{\baselinestretch}%
		{#1}\small\normalsize} \spacingset{1}

	\if0\blind
	{
		\title{\bf Optimal Transport Weights for Causal Inference}
		\date{}
		\author{Eric A. Dunipace\thanks{edunipace@mail.harvard.edu}\\
			\hspace{.2cm}\\
			David Geffen School of Medicine at UCLA\\
			Department of Biostatistics, Harvard T.H. Chan School of Public Health
		}
		\maketitle
	} \fi

	\if1\blind
	{
		\bigskip
		\bigskip
		\bigskip
		\begin{center}
			{\LARGE\bf Optimal Transport Weights for Causal Inference}
		\end{center}
		\medskip
	} \fi
	
	\bigskip

	\begin{abstract}
		Imbalance in covariate distributions leads to biased estimates of causal effects. Weighting methods attempt to correct this imbalance but rely on specifying models for the treatment assignment mechanism, which is unknown in observational studies. This leaves researchers to choose the proper weighting method and the appropriate covariate functions for these models without knowing the correct combination to achieve distributional balance. In response to these difficulties, we propose a nonparametric generalization of several other weighting schemes found in the literature: Causal Optimal Transport. This new method directly targets distributional balance by minimizing optimal transport distances between treatment and control groups or, more generally, between any source and target population. Our approach is semiparametrically efficient and model-free but can also incorporate moments or any other important functions of covariates that a researcher desires to balance. Moreover, our method can provide nonparametric estimate the conditional mean outcome function and we give rates for the convergence of this estimator. Moreover, we show how this method can provide nonparametric imputations of the missing potential outcomes and give rates of convergence for this estimator. We find that Causal Optimal Transport outperforms competitor methods when both the propensity score and outcome models are misspecified, indicating it is a robust alternative to common weighting methods. Finally, we demonstrate the utility of our method in an external control trial examining the effect of misoprostol versus oxytocin for the treatment of post-partum hemorrhage.
	\end{abstract}
	
	\noindent%
	\vfill
	
	\doublespacing
	
	\section{Introduction}
	Inverse propensity score weighting (IPW) uses the predicted probabilities of treatment to estimate causal effects. Under the correct model, these weights will lead to distributional balance between treatment groups and thereby to unbiased effect estimates. Unfortunately, adequate distributional balance is a challenge to determine, so researchers often assess performance by measuring the balance of observed covariate functions \citep{li_balancing_2018}. 
	
	Rather than using weights that indirectly balance selected covariate functions, other authors construct weights that achieve such balance by design 
	\citep{Hainmueller2012, Imai2014, Zubizarreta2015}; however, the particular aspect of the covariates that will achieve distributional balance, especially in finite samples, is usually unknown.
	Weights built using reproducing kernel Hilbert spaces (RKHS) may get around this difficulty because they theoretically balance all covariate functions  \citep{Li2021}. 
	Nonetheless, RKHS weights necessitate the tuning of unknown hyperparameters without an obvious metric to assess performance. Some authors tune the RKHS weights such that they balance predictions from an outcome model \citep{Kallus2016,Kallus2018}, 
	meaning the performance of the weights in turn depends on the performance of these outcome models.
	However, all of these balancing methods fail to capture something fundamental about the covariates: the joint distribution.
	%
	%
	
	Ultimately, if the distributions are the same between treatment groups, then all functions of the covariates will be balanced, leading to unbiased treatment effect estimates.
	But researchers may have certain covariate functions that they believe are important \textit{a priori}. As such, researchers may not want to trade off balance on these quantities for better overall distributional balance. 
	Fortunately, we do not have to choose between balancing distributions or covariate functions---we can do both. 
	
	Causal Optimal Transport (COT) is our proposed method that
	balances the joint distribution and any other selected covariate functions of interest in one framework.
	As we document, our method is semiparametrically efficient and performs well in simulation studies compared to competing methods---especially when the propensity score and outcome models are misspecified. 
	We also demonstrate that several methods in the literature are actually special cases of COT, meaning our framework can be seen as an interpolation between several seemingly unconnected methods. 
	Finally, we show how optimal transport methods can nonparametrically impute the missing potential outcomes by estimating the conditional mean outcome function at a $n^{-2/(2d' + 9)}$ rate, where $d' = \lceil 5d/4\rceil$ and $d$ is the dimension of the covariate space.
	\if0\blind{Code to implement the methods discussed in this work is available in the new \texttt{R} package \texttt{causalOT}, found at \url{http://www.github.com/ericdunipace/causalOT}.}\fi
	
	The paper proceeds as follows: in Section \ref{sec:setup}, we describe the setting and assumptions necessary for identification. Then we introduce optimal transport and COT in Section \ref{sec:otsec}. In Section \ref{sec:compare}, we demonstrate how COT unifies several existing methods, and in Section \ref{sec:sims}, we provide simulation results demonstrating the utility of COT. Section \ref{sec:case} presents a case study utilizing our method in a study of post-partum hemorrhage. Finally, we offer our concluding remarks in Section \ref{sec:summary}.
	
	\section{Setup}\label{sec:setup}

	\subsection{The potential outcomes framework}
	We adopt the potential outcomes framework of Neyman and Rubin \citep{Splawa-Neyman1923, Rubin1974}. Assume that we have an independent, identically distributed (iid) sample of $n \in \Nat$ units from some population. Let $Z$ be a binary variable that denotes receiving either a treatment ($Z = 1$) or control ($Z = 0$) condition. $Y(0)$ and $Y(1)$ are the potential outcomes, and $Y = Z \cdot Y(1) + (1-Z) \cdot Y(0)$ is the observed outcome defined on a space $\mcY \subseteq \R$. The confounders are $X \in \mcX \subseteq \R^d$. We will assume we have $n_0 \in \Nat$ control units and $n_1 \in \Nat $ treated units giving $n = n_0 + n_1$ total observations from this sample. Denote $\bolda_z = \sum_{i: Z_i = z}^n \delta_{x_i} a_i$ as the empirical distribution for treatment $Z = z$ and $\bolda = \sum_i^n \delta_{x_i} a_i$ as the empirical distribution for the full sample. Let $\alpha_z$ and $\alpha$ be the corresponding population distributions.

	Finally, we assume the space $\mcX$ has a distance metric between observations, $d_{\mcX}(x_i, x_j) \in \R_+$. We will define a generic cost function as $c(x_i, x_j) = d_{\mcX}(x_i, x_j) ^p$ with $p \geq 1$. As an example, if $d_\mcX$ is the Euclidean distance and $p = 2$, then $c$ is the squared-Euclidean distance. 
	From this function, 
	we then construct a pairwise cost matrix $\boldC \in \R_+^{n \times m}$ between each unit $i$ and $j$:
	$\boldC_{ij} = c(x_i, x_j).$
	
	\subsection{Causal estimands}
	
	There are several potential causal contrasts of interest over these populations but we focus on the sample average treatment effect (ATE):
	\begin{equation}
		\tau = \E\left\{Y(1) - Y(0)\right\},
		\label{eq:ate}
	\end{equation}
	Unfortunately, we cannot estimate Eq. \eqref{eq:ate} since one of the potential outcomes is missing for each individual. 
	
	Instead, we need to use the information in the source population to get valid treatment effect estimates. That is, we desire $\E\left\{Y(z) \indicator(Z = z) \cdot w\right\} = \E\left\{Y(z)\right\},$ for some function $w$.
	A common way to do this is to use an importance sampling weight $w = \frac{d\alpha}{d\alpha_z}$, which is also known as the Radon-Nikodym derivative of $\alpha$ with respect to $\alpha_z$.
	
	With a known $w$, we estimate Eq. \eqref{eq:ate} as
	\begin{equation}
		\hat{\tau} = \sum_{i} w_i Y_i Z_i - \sum_{i} w_i Y_i (1 - Z_i ),
		\label{eq:hajek_init}
	\end{equation}
	and such that the weights sum to one in the treatment and control groups: $\sum_i w_i Z_i  = \sum_i w_i (1 -Z_i)  = 1$. 
	
	\subsection{Identifying assumptions}
	To identify these estimators, we need several assumptions, which we formalize below. 
	
	\begin{assumption}[] 
		Stable unit treatment value assumption, \citep{Rubin1986}:\\ $Y_i(Z_1, Z_2,..., Z_i,...,Z_n) = Y_i(Z_i)$
		and  $Y_i(Z_i) = Y_i(z) $ if  $ Z_i = z$.
		\label{assum:sutva}
	\end{assumption}
	\begin{assumption}[]
		Strong ignorability of treatment assignment, \citep{Rosenbaum1983}:\\ $Y(0), Y(1) \independent Z \given X, S = 1$ and $0 < P(Z=1 \given X, S = 1) < 1$.
		\label{assum:si}
	\end{assumption}
	These standard assumptions allow us to use the observed data to estimate the desired treatment effects in the target sample. Assumption \ref{assum:sutva} allows us to use the observed outcomes and not consider interference between units, while Assumption \ref{assum:si} gives common support between treatment populations. 
	With these conditions, we now turn to optimal transport and COT.

	\section{Causal Optimal Transport} \label{sec:otsec}
	
	\subsection{General properties of optimal transport}
	The popularity of optimal transport methods have exploded in recent years thanks to several recent theoretical and methodological advances \citep{Cuturi2013,Peyre2019}, but the field dates back centuries. 
	We frame our discussion in terms of empirical samples $\bolda_z$ and $\bolda$ to align with the rest of the paper but these quantities can be arbitrary samples for general optimal transport problems.
	
	The original optimal transport problem formulated by \citet{Monge1781}  involves finding optimal maps between distributions. Define such a map as a function $T: \mcX \mapsto \mcX$ and such that $\int_{\mcX } g(x) d\alpha_z = \int_{\mcX} g(T(x)) d\alpha$ for all measurable functions $g$. 
	We denote the corresponding push-forward operator from $\alpha$ to $\alpha_z$ as $T_\# \alpha = \alpha_z$. The \citeauthor{Monge1781} formulation of the optimal transport problem is then 
	\begin{equation}
		\inf_T \quad   \sum_{i} c\{T(x_i),x_i\} a_j,
		\label{eq:monge_prob}
	\end{equation}
	where $T_\# \alpha = \alpha_z$. 
	Unfortunately, this problem can be intractable to solve in practice since the mapping must be injective.
	
	To alleviate this issue, the \citet{Kantorovich1942} formulation instead considers probabilistic assignments between distributions $\bolda_z$ and $\bolda$ using a transport matrix $\boldP$:
	\begin{equation}
		\ot{\bolda_z}{\bolda}{}	 = \min_{\mathbf{P} \in \boldU(\bolda_z, \bolda)} \quad   \sum_{i:Z_i = z,\,j} \boldC_{ij}\boldP_{ij}    
		\label{eq:ot}
	\end{equation}
	where $\boldU(\bolda_z, \bolda)$ is the set of joint distributions with margins $\bolda_z$ and $\bolda$.  This metric is a proper distance that obeys the triangle inequality and  metrizes the convergence in distribution, \textit{i.e.} $\ot{\bolda_z}{\bolda}{} = 0 \iff \bolda_z = \bolda$ (Proposition 2.3, \citealp{Peyre2019}).
	When $c(x, x') = d_{\mcX}(x, x')^p$, as is the case for our setting, then Eq. \eqref{eq:ot} is also known as the $p\,$-Wasserstein distance.
	Unfortunately, this problem is known to have a decaying convergence with increasing dimension \citep{Weed2017} and also to suffer from a large computational complexity \citep{Cuturi2013}.		
	
	Conveniently, regularized optimal transport offers improved rates of asymptotic convergence \citep{Genevay2018, Mena2019}  and computational speed \citep{Altschuler2017} by adding a convex penalty to the objective function:
	\begin{equation}
		\ot{\bolda_z}{\bolda}{\lambda} = \min_{\mathbf{P} \in \boldU(\bolda_z, \bolda)} \sum_{i: Z_i = z,\,j} \boldC_{ij}\boldP_{ij} +  H_\lambda(\boldP_{ij}).
		\label{eq:ot_reg}
	\end{equation}
	Common penalties for $H_\lambda$ include an entropy penalty, $\lambda \boldP_{ij} \log \boldP_{ij}$ \citep{Cuturi2013}, or an $L_2$ penalty, $\frac{\lambda}{2}\boldP_{ij}^2$ \citep{Blondel2018}.
	The solutions to this problem converge to the solutions from Eq. \eqref{eq:ot} as $\lambda \to 0$, while as $\lambda \to \infty$, the solutions put equal weight on every entry in  $\boldP$.
	
	To adjust for the fact that $\ot{\bolda_z}{\bolda}{\lambda} \neq 0$, \citet{genevay_learning_2018} introduced the Sinkhorn divergence for entropy penalized optimal transport: 
	\begin{equation}
		S_\lambda (\bolda_z, \bolda) = \ot{\bolda_z}{\bolda}{\lambda} - \frac{1}{2} \ot{\bolda_z}{\bolda_z}{\lambda}
		- \frac{1}{2} \ot{\bolda}{\bolda}{\lambda}
		\label{eq:sink_div}
	\end{equation}
	This has the advantage that $S_\lambda (\bolda_z, \bolda) = 0 \iff \bolda_z = \bolda$ \citep{Feydy2018}, while retaining the computational and theoretical advantages of regularized optimal transport.

	Finally, we can still use Eqs. \eqref{eq:ot} or \eqref{eq:ot_reg} to construct a map as in Eq. \eqref{eq:monge_prob}.
	In finite samples, this function can be estimated from the Kantorovich formulation as
	\begin{equation}
		T_{\bolda \mapsto \bolda_z}(j) = \argmin_\upsilon \sum_{i: Z_i = z} c(X_i, \upsilon)\boldP_{ij}.
		\label{eq:map}
	\end{equation}
	This mapping is alternatively known as the barycentric projection \citep{Peyre2019}.
	For the squared-Euclidean cost, this map equals $\frac{1}{a_j}\sum_{i: Z_i = z} \boldP_{ij} X_i$,
	or the weighted mean of the observations in the sample who received treatment $Z=z$. For an $L_1$ cost, $T_{\bolda \mapsto \bolda_z}$ is the weighted median of the corresponding $X_i$. 
	Under an $L_2$ cost, this map will also converge to the optimal Monge map provided one of underlying measures is continuous \citep{Ambrosio2005}. With these general properties established, we now turn to our proposed method.
	
	\subsection{Problem formulation}\label{sec:otw_sec}
	We define the COT problem as
	\begin{align}
		\cotobj{\bolda}{\lambda} &= \min_{\boldw \in \Delta_n} \; S_\lambda(\boldw_1,\bolda) + S_\lambda(\boldw_0,\bolda),
		\label{eq:COT_div}
	\end{align}
	where $\boldw_z$ is the empirical measure $\sum_{i: Z_i = z} \delta_{x_i} w_i $, $\Delta_n$ is the simplex with $n$ vertices, and $S_\lambda$ is defined in \eqref{eq:sink_div}.  The COT weights will then be the weights that minimize $\cotobj{\bolda}{\lambda}$. In a slight abuse of notation, we have the following marginal distribution:
	\begin{equation}
		\boldw_{\text{COT}} = \argmin_{\boldw \in \Delta_n} \quad \cotobj{\bolda}{\lambda}.
		\label{eq:margOT}
	\end{equation}
	
	In addition to seeking distributional balance, a researcher may also know a set of functions that he or she thinks are important to balance \textit{a priori} for valid causal estimates. These functions may include a hypothesized outcome model or the moments of the covariates. Define $B_k(\cdot): \mcX \mapsto \R$ for $k \in \{1,...,K\}$ as these $K$ functions of interest. We can then add an additional constraint to the problem in Eq. \eqref{eq:COT_div} to approximately balance these important functions between samples:
	\begin{equation}
		\left|\sum_{i: Z_i = z} B_k(X_i) w_i -  \frac{1}{m} \sum_{j} B_k(X_{j}) \right|  \leq \delta_k, \, \, \forall k \in \{1,...,K\}. \label{eq:cot_bf}
	\end{equation}
	
	Of course, other formulations of the problem are possible and we detail some of them in Appendix \ref{sec:other_ot} of the Supplementary Materials. However, we find that in practice the formulation in Eq. \eqref{eq:COT_div} has the best performance in terms of bias and variance.

	\subsection{Convergence}
	We now discuss the convergence of our weights to the distribution of interest. First, we define the importance sampling weights as $\breve{w}_i^\star = \frac{d\alpha(X_i)}{d\alpha_z(X_i)}$ and define the self-normalized importance sampling weights as $w_i^\star =  \frac{1}{n}\breve{w}_i^\star/\sum_{i} \frac{1}{n}\breve{w}_i^\star.$
	\noindent In our setting, $\breve{w}_i^\star =1/$ $\p(Z_i = z \given X_i) $.
	Further, let $\delta_{n}$ be the smallest value of the balancing function constraints at which the importance sampling weights satisfy the condition in Eq. \eqref{eq:cot_bf} for sample size $n$. We also rely on some additional assumptions to prove the convergence of the COT weights.
	
		%
	\begin{assumption}
		\label{assum:conv}
		$\exists x_0 \in \mcX: \int_\mcX c(x_0, x)d\alpha < \infty$ and $\E_\alpha|B(X)| < \infty$ with $\|\delta_{n}\|^2 = o_p(\|\delta\|^2)$.
	\end{assumption}
	\begin{assumption}
		$c(\cdot,\cdot) $ is in $\mcC^\infty$ and is $L$-Lipschitz and either 1) $ \frac{1}{\lambda^{\lceil 5 d/4 \rceil + 2}} \frac{1}{\sqrt{n}}  = o_p(1)$ and $\alpha_z$ and $\alpha$ are $\sigma^2$-subgaussian with $c = \|\cdot\|_2^2$  or 2) $\frac{\exp(\|\boldC\|_\infty / \lambda)}{\lambda^{\lceil d/2 \rceil}} \frac{1}{\sqrt{n}}  = o_p(1)$ and $\mcX \subset \R^d$.
		\label{assum:conv_rootn}
	\end{assumption}
	These assumptions enforce some regularity on the constituent parts of COT. First, the cost function must exist and be continuously differentiable. Second, the measures are either subgaussian or defined on subsets of the real numbers. Third, the penalty term $\lambda$ cannot go to zero too quickly if at all. And fourth, if using balancing constraints, there needs to be a value at which the importance sampling weights satisfy the constraints. With these assumptions, we have our first theorem.
	
	\begin{theorem}\label{thm:conv_cot}
		If Assumptions \ref{assum:si}--\ref{assum:conv_rootn} hold,
		then as $n \to \infty$,
		\[
		\boldw_{\text{COT}} \rightharpoonup \alpha,
		\]
		where $\boldw_{\text{COT}}$ is defined in Eq. \eqref{eq:margOT}. Further, 
		\[\E\left\{ \ot{\boldw_{\text{COT}}}{\bolda}{\lambda} - \ot{\alpha}{\alpha}{\lambda}  \right\} = \mcO \left(\frac{1}{\sqrt{n}}\right).\]
	\end{theorem}
	%
	This theorem says that the COT weights converge to the distribution of the target sample at a $\sqrt{n}$-rate, which also has implications for the efficiency of estimators based on COT, as we will see in the next section. 
	A proof of this theorem is provided in Appendix  \ref{sec:conv_proof} of the Supplementary Materials.

	Finally, Theorem \ref{thm:conv_cot} also gives the following corollary.
	\begin{corollary} 
		\label{coro:cot_conv_is}
		As $n \to \infty$,
		\[\lim_{n \to \infty} \boldw_{\text{COT}} \stackrel{\text{a.s.}}{=} \lim_{n \to \infty} \boldw^\star.\]
	\end{corollary}
	\noindent The corollary follows as a consequence of the Radon-Nikodym Theorem and the fact that the Radon-Nikodym derivatives are almost surely unique.
	%
	%
	%
	
	\subsection{Statistical Inference}
	
	For statistical inference, we turn our attention to the asymptotic distribution of Eq. \eqref{eq:hajek_init} and its variance.
	We assume the following conditions also hold.
	
	\begin{assumption}
		\label{assum:outcome_reg}
		$\E| Y - \mu_z(X)| < \infty$ for $\mu_z(X) \defn \E\{Y(z) \given X\}$, $\E(Y^2) < \infty$, and either $S_\lambda (\boldw_{\text{COT}}, \bolda) = o_p(n^{-1/2})$ or, for basis function balancing, $\| \delta \|_2^2 =  o_p(n^{-1/2})$ with $\mu_z(X) \subseteq B(X)\trans \gamma $ for $\gamma \in \R^K$.
	\end{assumption}
	
	This assumption has several important parts. We assume that the second moment of the outcome is finite and that the residual is $L_1$-integrable, which are not strong assumptions for real data. Then we require one of two additional conditions to hold. 
	The first potential condition is that the convergence of $\boldw_{\text{COT}}$ to $\bolda$ occurs at a faster than $\sqrt{n}$-rate. We note that this is not actually that strong of an assumption in practice since the convergence to $\alpha$ happens at a $\sqrt{n}$-rate and COT is directly targeting the empirical distribution $\bolda$. Thus, we expect the convergence to the empirical distribution to be faster than $1/\sqrt{n}$, which is what we  observe in practice. The second potential condition requires that the basis functions $B$ encompass the true conditional mean and that the empirical means of $B$ converge faster than $1/\sqrt{n}$. We observe in practice that the convergence of the basis functions is actually possible with relatively small sample sizes, making this assumption very plausible; however, outcome models are typically not known---though this can be ameliorated by using nonparametric  models. With these conditions, we have our next theorem.
	\begin{theorem}
		\label{thm:vopt}
		If Assumptions \ref{assum:sutva}--\ref{assum:outcome_reg} hold, then as $n \to \infty$, 
		\[ \sqrt{n} \left (\hat{\tau} - \tau \right) \dist \N\left(0, V_{\text{opt}}\right ), \]
		where $V_{\text{opt}}$ is the semiparametrically efficient variance as in Theorem 1 of \citet{Hahn1998}.
	\end{theorem}
	\noindent This result follows from the fact that the expansion of the bias $\hat{\tau} - \tau$ has the form of the semiparametrically efficient score function.
	We defer a proof to Appendix \ref{sec:vopt_proof} of the Supplementary Materials. 
	
	Theorem \ref{thm:vopt} also gives us the following corollary.
	\begin{corollary} \label{coro:dr}
		Under Assumptions \ref{assum:sutva}--\ref{assum:outcome_reg}, then COT is doubly robust for a large enough $n$:
		\begin{align*}
			\hat{\tau}_{dr} & = n \inv\sum_{i=1}^n  n w_i Z_i \{Y_i - \mu_1(X_i) \} - n w_i (1-Z_i) \{Y_i - \mu_0(X_i) \}  + \{\mu_1(X_i) - \mu_0(X_i)\} \\
			&\approx  n \inv\sum_{i=1}^n  n w_i Z_i Y_i -  n w_i (1-Z_i) Y_i = \hat{\tau}
		\end{align*}
	\end{corollary}
	\noindent In practice, one can check if Corollary \ref{coro:dr} holds by examining both the optimal transport distance between distributions $\boldw_{\text{COT}}$ and $\bolda$ as well as the balance of the hypothesized outcome models between samples of the covariate functions that determine the assumed outcome models . If the hypothesized outcome models are well-balanced, then there is little benefit to model augmentation. A manifestation of this phenomenon can be seen in the simulations in Section \ref{sec:sims} where adding in model augmentation does not change the estimates from using COT even for sample sizes as low as $500$. 

	\subsection{Imputing the missing potential outcomes}
	One of the advantages of COT is that it provides a method to impute the missing potential outcomes, if so desired. Moreover, methods based on this estimator will also converge to the correct treatment effect, even when using weights from other methods. 
	
	We can construct a transportation matrix, $\boldP$, \textit{a posteriori} for Eq. \eqref{eq:COT_div} by solving $\ot{\boldw_z}{\bolda}{\lambda'},$ for any $\lambda' > 0$ using the appropriate weights from treatment group $Z=z$. Then the missing potential outcomes can be estimated by the barycentric projection in Eq. \eqref{eq:map}: $\hat{Y}_j(z) =  \argmin_\nu \sum_i \indicator(Z_i = z) c(Y_i, \nu) \boldP_{ij}$. In practice, we do not have to use the same cost function used to estimate $\boldP$ but using a squared-$L_2$ cost gives us the following theorem
	
	\begin{theorem}\label{thm:bp_conv}
		Assume that $\alpha$ is compactly supported and admits a density with finite Fisher information $I_0$ and finite second moments, and $c=\| \cdot \|_2^2$. Further, assume $\mu_z$ is $L$-Lipschitz, $\var(Y - \mu_z \given X)< \xi^2 < \infty$ for all $X \in \mcX$, and that Assumptions \ref{assum:sutva}--\ref{assum:si} hold.
		Then if $\boldw \rightharpoonup \alpha$, $\boldP$ is estimated via Eq. \eqref{eq:ot_reg} with an entropy penalty, and $\lambda \asymp n^{- \frac{1}{2 d' + 9}}$ for $d' = \lceil 5d/4 \rceil$, 
		\[\E_\alpha \|\hat{Y}(z) - \mu_z(X) \|^2 \lesssim (1 + I_0) n^{- \frac{2}{2d' + 9 }}.\] 
	\end{theorem}
	There are several things to observe about this theorem and its assumptions. 
	First, the assumptions on the cost and distribution allow us to connect the optimal transport solutions to the Monge maps of Eq. \eqref{eq:monge_prob}, but also allow us to give a rate for our theorem. Second, using penalized optimal transport ensures that we average out the errors $Y - \mu_z$. Third, in a related manner, the Lipschitz continuity of the outcome means that this averaging out of the errors will still achieve good estimates of $\mu_z$. Fourth and finally, this theorem also suggests that $\mu_z$ is like a Monge map, $T$, between potential outcomes.
	
	Unfortunately, these imputations are not necessarily useful by themselves.
	\begin{proposition}
		Under a squared-$L_2$ metric, an ATE estimator based solely on the barycentric projection, $n \inv\sum_i\hat{Y}_i(1) - \hat{Y}_i(0)$,
		is equivalent to Eq. \eqref{eq:hajek_init}.
		\label{prop:l2_map}
	\end{proposition}
	\noindent Clearly, care must be taken when using these estimators. Proofs are located in Appendix \ref{sec:bp_proof} and \ref{prop:l2_map} of the Supplementary Materials.
	
	\subsection{Practical considerations}
	
	In this section, we turn to the practical considerations of optimizing the COT weights. Namely, we discuss the tuning of the hyperparameters and estimation of the weights.
	
	\paragraph{Hyperparameter tuning.}
	Our goal is to select the hyperparameters $\lambda$ and $\delta$ so that we achieve the best distributional balance without over-fitting the current data.
	To do so, we propose a bootstrap based tuning procedure detailed in Algorithm \ref{alg:bootwass}.
	
	\begin{algorithm}[tb]
		
		\KwData{Grid of parameter values $\Delta = \{\{\lambda_1,  \delta_1\},  \{\lambda_2,  \delta_2\},...\}$, number of bootstrap samples $K$, empirical measure $\bolda$,  empirical measure $\bolda_z$, treatment group of interest $z$, $\lambda' \geq 0$}
		\KwResult{Value of hyperparameters, $\hat{\lambda}, \hat{\delta}$ }
		\ForEach{$\{\lambda, \delta\}\in \Delta$}{
			Estimate weights $\boldw_{\text{COT}}$ given parameters $\{\lambda, \delta\}$\;
			\For{$k$ in $1,...,K$}{
				Bootstrap new target data  $\bolda_k ^\star \sim \bolda$\;
				Bootstrap new source data $\bolda_{z,k} ^\star \sim \bolda_z$\;
				Set the unnormalized weights for empirical measure in treatment group  $z$ as $\tilde{\boldw}_k^\star  = \indicator(Z=z) \odot \boldw \odot \bolda_{z,k}^\star $, where $\odot$ is the element-wise product\;
				Renormalize the weights, $\boldw_k^\star = \tilde{\boldw}_k^\star / (\tilde{\boldw}_k^{\star \top} \ones_n)$\;
			}
			Set $T_{\{\lambda,\delta\}} =  K\inv\sum_k^K \ot{\boldw_k ^\star}{\bolda_k ^\star}{\lambda'}$\;
		}
		\Return $\hat{\lambda}, \hat{\delta} = \argmin_{\{\lambda, \delta\}}  T_{\{\lambda,  \delta\}}$\;
		\caption{Choosing hyperparameters for COT}
		\label{alg:bootwass}
	\end{algorithm}

	We justify this procedure in two ways. First, practitioners probably do not have an ideal weight penalty in mind based on subject matter knowledge. Second, because the COT weights target the Radon-Nikodym derivative, 
	this tuning procedure will select the hyperparameters that lead to weights robust to sampling variation and better approximation of these population level quantities. We present an empirical examination of this tuning algorithm in Appendix \ref{sec:tun_alg} of the Supplementary Materials that demonstrates its effectiveness at selecting the optimal $\lambda$.
	
	\paragraph{Weight estimation.}
	Given the known complexity of estimating optimal transport distances, \citet{huling_energy_2020} raise the concern that methodology like COT will not be feasible. Fortunately, these concerns are addressed by using regularized optimal transport.
	
	In our simulations, we find that estimating the COT weights only takes a few seconds for a 1000 observations. Eq. \eqref{eq:COT_div} can be solved by alternating Sinkhorn divergence calculations in GeomLoss \citep{Feydy2018} and optimization steps on the weights. With balancing constraints, we use the Frank-Wolfe algorithm to optimize the weights \citep{Frank1956}; without balancing constraints, we can use an LBFGS algorithm.
	
	\subsection{Target average treatment effects and multi-valued treatments}
	COT is also well-suited to the case where the target estimand is for a separate sample entirely. This is because the weights can be calibrated to target any arbitrary set of samples. The only additional assumption for convergence is that there is common support between distributions.
	
	Finally, COT is easily extended to more than two treatments as long as the treatment values are discrete. This is because COT estimates  weights separately for each treatment group.

	\section{Connections to Existing Methods} \label{sec:compare}
	The COT framework is actually related to several other methods in the literature and can be seen as an interpolation between all of them, as we detail below.
	
	\paragraph{Synthetic control method}
	\citet{Abadie2003} first proposed the synthetic control method (SCM) as a way of performing counterfactual inference for a single treated unit, $j$. 
	The objective function is 
	\[\min_{w: w \trans \ones_n = 1, w_i \geq 0} \left \| \sum_i X_i w_i - X_j \right \|_2^2,\]
	which is the same objective as the Monge map in Eq. \eqref{eq:monge_prob} when $c$ is the squared-Euclidean distance and $T$ has the corresponding form of the barycentric projection in Eq. \eqref{eq:map}. This means that SCM is actually estimating a version of the COT problem with $\lambda = 0$.
	\begin{proposition} \label{prop:scm}
		If $\alpha$ admits a density, $c=\| \cdot \|_2^2$, and $\exists x_o \text{ s.t. } \E\{c(x_0,X)\} < \infty$,
		then SCM is asymptotically the same as COT with $\lambda = 0$.
	\end{proposition}
	\noindent For a proof, see Appendix \ref{sec:scm_proof} in the supplementary materials.
	One potential drawback of using SCM versus the formulation used for COT can be seen in the following simple example in Figure \ref{fig:scm_counter}. COT favors the nearest point while the SCM method utilizes the points further away, which could be a problem if the response surface looks like Figure \ref{fig:scm_response_surf}. To avoid this, SCM could incorporate a modified objective that directly models both the barycentric projection and distance between units as in \citet{Perrot2016}.
	
	\begin{figure}[tb]
		\begin{subfigure}[b]{0.4\textwidth}
			\centering
			\includegraphics[width=\textwidth]{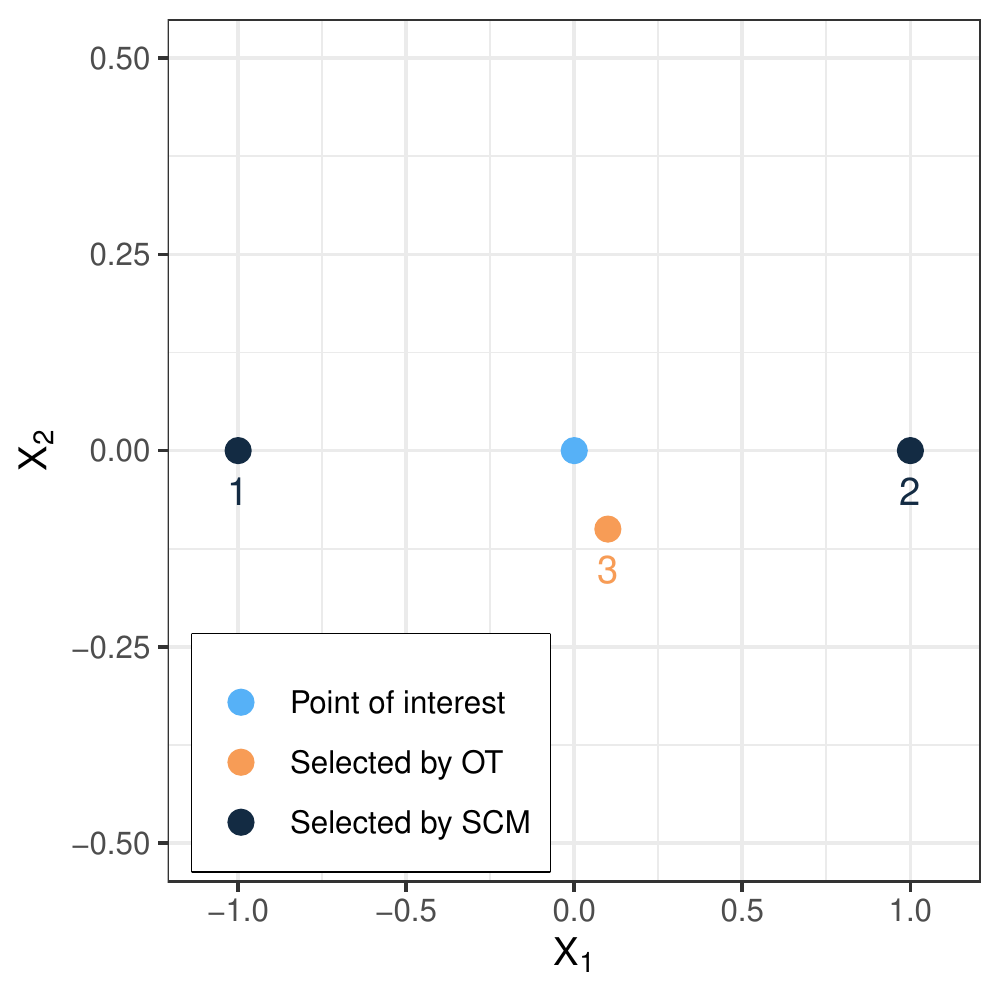}
			\caption{Potential points for selection and the point of interest.}
		\end{subfigure}
		\begin{subfigure}[b]{0.6\textwidth}
			\centering
			\includegraphics[width=\textwidth]{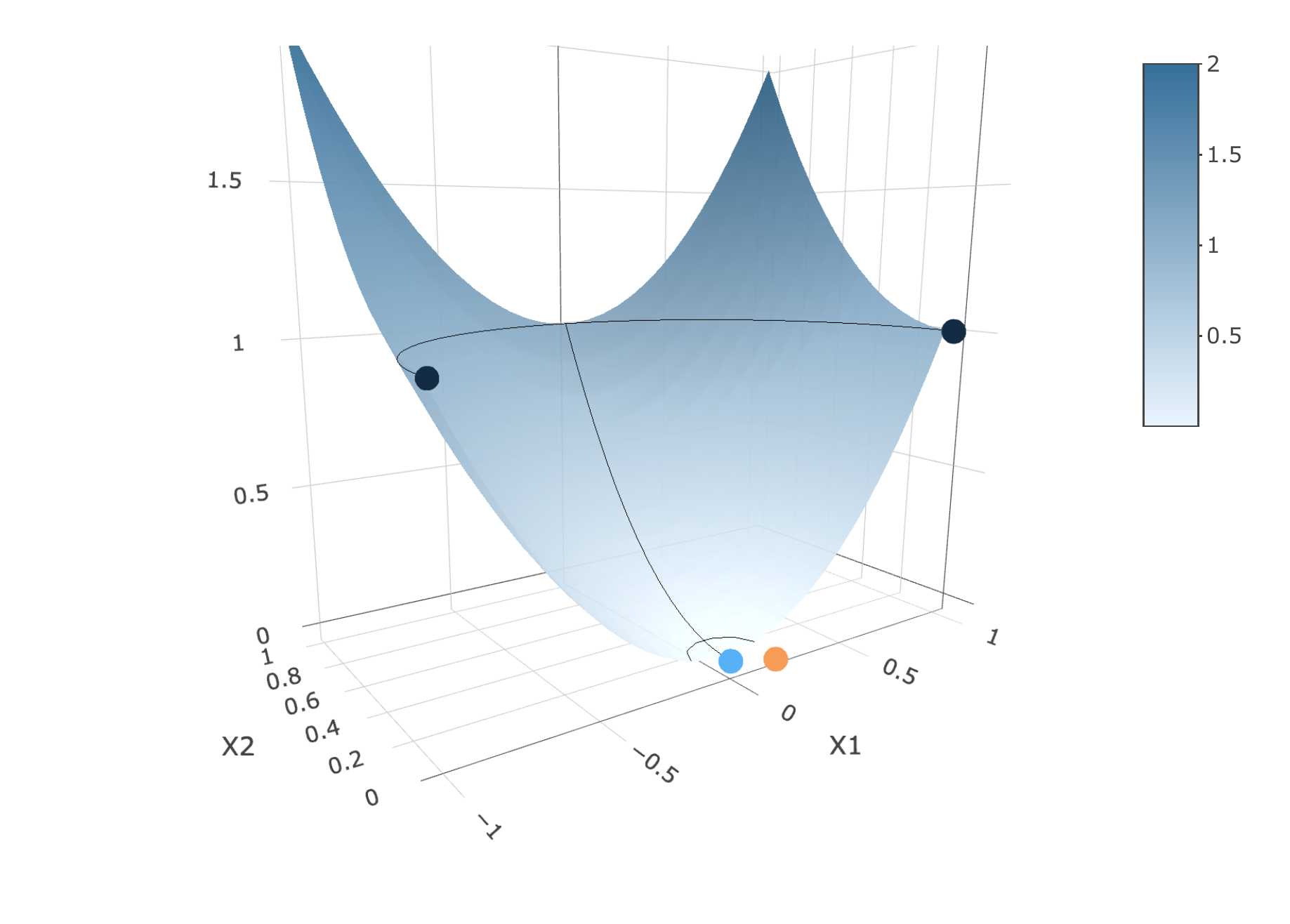}
			\caption{A response surface where synthetic controls would lead to poor estimates.}
			\label{fig:scm_response_surf}
		\end{subfigure}
		\caption{A simple example where synthetic controls (SCM) would perform worse than COT.}
		\label{fig:scm_counter}
	\end{figure}
	
	\paragraph{Nearest neighbor matching}
	When done with replacement, nearest neighbor matching (NNM) is also a reformulation of the COT problem with $\lambda=0$:
	$\min_{\mathbf{P} \geq 0} \;  \sum_{i: Z_i = z, j} \boldC_{ij}\boldP_{ij}$ subject to $\sum_{i,j}  \boldP_{ij} \indicator(Z_i = z)  = 1$ and $\boldP \trans \ones_{n} = \bolda.$
	This will seek to find the unit $i$ that is closest in terms of $\boldC$ for each unit $j$ since this minimizes the total cost, which is the definition of NNM with replacement. Each observation will simply be weighted by the number of times it is matched, divided by the total sample size $n$.
	
	NNM has some advantages and drawbacks relative to more general COT. Positive aspects of the method are that quick to estimate and also corresponds to an easily understood quantity of physical matching familiar to many researchers. However, the method can have poor convergence properties if $c \neq d_\mcX(\cdot, \cdot)^p$ where $p > d/2$ \citep{Fournier2015}. This means that additional assumptions are necessary to ensure adequate convergence. Further, the weights will be given as rational numbers and as such, we would expect estimators based on them to have higher variance than weights without such constraints.
	
	\paragraph{Optimal matching and MIP matching}
	Optimal Matching \citep{Rosenbaum1989} corresponds to COT weights with integer solutions and $\lambda = 0$:
	$\min_{\mathbf{P} \in \{0,1\}}   \sum_{i: Z_i = z, \,j} \boldC_{ij}\boldP_{ij}$ subject to $L \leq \sum_{i}  \boldP_{ij} \indicator(Z_i = z)  \leq U, \; \forall j$ and $1 \leq L \leq U \leq n.$ We can turn this into Mixed Integer Program (MIP) matching by adding constraints: \[\left | \sum_{ij} \frac{B_k(X_i) \boldP_{ij}\indicator(Z_i = z)}{n M_j } - \frac{1}{n} \sum_{j}  B_k(X_j) \right| \leq \delta_k, \, \, \forall k \in \{1,...,K\},\] where $M_j = \sum_{i} \boldP_{ij}$ is the number of matches for unit $j$ \citep{Zubizarreta2012}.
	These methods also imply a re-weighting of treatment group $Z=z$ since the weights on element $i$ will be $n \inv \sum_j \boldP_{ij} /M_j$. 
	
	These methods have a similar flavor to NNM but with additional linear constraints. As such, they would likely share some of its benefits and drawbacks. As an advantage, these methods again yields matches which have an easy interpretation; however, as a disadvantage, the additional linear constraints will slow down the problem estimation. Moreover, the problem has additional tuning parameters $L$, $U$, and $\delta_k$ not present in NNM. Similar to NNM, we would again expect weights based on rational numbers to have higher variance.


	\paragraph{Energy distance.}
	The Energy Distance (ED) is defined as 
	$\mcE(\bolda_z, \bolda) = \frac{2}{n^2} \sum_{i:Z_i = z, \,j} d_\mcX(x_i, x_j)^p$
	$- \frac{1}{n^2} d_\mcX(x_i, x_i)^p
	- \frac{1}{m^2} d_\mcX(x_j, x_j)^p,$
	where $p \geq 1$. Then $S_\lambda(\bolda_z, \bolda)\to \frac{1}{2} \mcE(\bolda_z, \bolda)$ as $\lambda \to \infty$ \citep{Feydy2018}. Thus, Energy Balancing Weights \citep{huling_energy_2020} are a special case of COT.
	
	One advantage of this method is that there is no tuning parameter necessary to estimate the weights. However, we might assume that there would be an advantage to interpolating between all of these various methods. Indeed, in our experiments we find that the larger values of $\lambda$ do not approximate the true inverse propensity score as well as intermediate values. See Appendix \ref{sec:tun_alg} in the Supplementary Materials for an empirical evaluation.

	\paragraph{Mean Maximum Discrepancy.} Optimal transport is also related to the mean maximum discrepancy (MMD) through the ED. The MMD is equal to $\mcM = 0.5 \int_{\mcX \times \mcX} k(x, x') d \phi(x),$
	for $\phi = \alpha - \beta$. For some reproducing kernel Hilbert space $k$ and for a distance $d_\mcX$ defined as 
	$d_\mcX(x, x') = \frac{1}{2} k(x, x') + \frac{1}{2} k(x, x') - k(x, x')$,  MMD is equivalent to the ED \citep{Feydy2018} and, therefore, to COT. To our knowledge, there has not been a proposed weighting method based on the MMD but we would expect it to have performance similar to that of the ED and COT with large values of $\lambda$.
	
	\section{Simulation Study}\label{sec:sims}
	To evaluate the finite sample performance of the proposed weighting methodology, we use the simulation study originally presented in \citet{Hainmueller2012}. For each setting, we run 1000 experiments with a sample size of $n = 512$. The estimand of interest is the ATE.
	
	For additional experiments examining the convergence of COT, its confidence interval coverage, and the efficacy of Algorithm \ref{alg:bootwass}, see Appendix \ref{sec:emp_studies} of the Supplementary Materials.
	
	\subsection{Setup}
	\paragraph{Study design.}
	We generate six covariates $X_1, ..., X_6$ from the following distributions
	\begin{align*}
		\begin{bmatrix} X_1 \\ X_2 \\ X_3 \end{bmatrix} &\sim \N \left( \begin{bmatrix} 0 \\ 0 \\ 0 \end{bmatrix}, \begin{bmatrix} 2 & 1 & -1 \\  1 & 1 & -0.5 \\ -1 & -0.5 & 1 \end{bmatrix} \right)\\
		X_4 &\sim \Un(-3,3)\\
		X_5 &\sim \chi^2_1\\
		X_6 &\sim \Bern{0.5}.
	\end{align*}
	In this study, the last three covariates are mutually independent of each other and also of the first three covariates.
	
	The treatment indicator is generated as $Z = \indicator(X_1 + 2 X_2 - 2 X_3 - X_4 - 0.5 X_5 + X_6 + \nu > 0),$
	where $\nu$ is drawn from one of three distributions leading to different degrees of overlap:
	$\nu \sim \N(0,100)$ for high overlap, $\nu \sim g(\chi_5^2)$ for medium overlap, and $\N(0,30)$ for low overlap.
	The function $g$ in the medium-overlap setting gives the $\chi_5^2$ draws expectation 0.5 and variance 67.6. 
	We expect the scenarios that will lead to the highest bias to be in the low-overlap setting where there is a strong separation between treatment groups and in the medium-overlap setting where the errors are leptokurtic.
	
	Given the treatment indicator $Z $ and the covariates $X_1,...,X_6$, we draw the outcome $Y$ from $Y(0) = Y(1) = (X_1 + X_2 +X_5 )^2 + \eta,$
	with $\eta \sim \N(0,1)$. 
	There are two things to note about this outcome model. The first is  that there is no effect of the treatment at the unit level and hence the ATE is 0. The second is that a linear outcome model should be biased.
	
	\paragraph{Methods under examination.}
	We compare our methodology to several other weighting methods commonly used in the literature. The first method we consider is a logistic regression (GLM) using only first order terms. We also consider balancing methods such as the covariate balancing propensity score (CBPS) of \citet{Imai2014} and the stable balancing weights (SBW) of \citet{Zubizarreta2015} both targeting mean balance. Finally, we also utilize SCM and NNM.
	
	For the 
	COT 
	weights, we include two variations both using an $L_2$ metric with standardized covariates. The first only balances the joint distribution (no constraints or ``none'') and the second demonstrates basis function balancing by targeting the joint distribution as well as mean balance (``means''). 
	
	\subsection{Estimators} \label{sec:estimators}
	In our simulations, we consider three estimators to target the ATE. 
	The first is known as the H\'{a}jek estimator \citep{Hajek1988} and is simply a weighted mean with sum to one weights as in Eq. \eqref{eq:hajek_init}.
	The second is an augmented or doubly robust estimator of \citet{Robins1994}, and the third is a weighted least squares estimator---both only including linear terms of the covariates.
	We do not include a barycentric projection estimator like in Eq. \eqref{eq:map} since under an $L_2$ metric it gives the same result as the H\'ajek estimator (see Proposition \ref{prop:l2_map}).

	\subsection{Results}
	
	
	We now turn our attention to the results. Due to its nonparametric to semiparametric formulation, we expect COT to do better than other methods when the true propensity score model diverges from a logistic regression---\textit{e.g.}, in the medium-overlap scenario. 
	
	Indeed, the COT methods have the lowest RMSE across all overlap scenarios, as we can see in Table \ref{tab:hain}. Further, across the medium and low overlap settings, COT has the lowest bias as well; in the high overlap scenario, there is a negligible difference between COT, SBW, and GLM.
	
	We should also note that COT gives estimates that do not vary between estimators. This is because the COT already balances the basis functions used in the augmented estimator and weighted least squares. Therefore, there is no difference between running an outcome regression utilizing linear terms of the covariates and the H\'{a}jek estimator. This is a function of the fact that by balancing distributions of the covariates, COT will also balance functions of the covariates. Similar phenomenon can also be observed with SBW.
	
	\input{tables/hainmueller.tex}

	\section{Case Study}\label{sec:case}
	In this section, we apply our methodology to a real data set. There is growing interest in the literature to utilize libraries of 
	randomized control trials (RCTs) to evaluate new interventions, the idea being that running new RCTs to evaluate every new intervention is expensive and time consuming \citep{Schmidli2020}. These studies, alternatively called externally controlled trials or synthetic control group trials, compare a set of study subjects receiving a treatment to a group of individuals external to the trial at hand who did not receive the intervention of interest. The participants used for the control group can be taken from a variety of sources such as an observational study, electronic medical records, or from historical clinical trial data \citep{Davi2020}.  
	To demonstrate this in practice, we present an analysis utilizing data originally from a multi-site RCT  discussed by \citet{Blum2010}. 
	
	\subsection{Misoprostol for the Treatment of Postpartum Hemorrhage}
	The original study was a double-blind, non-inferiority trial that exposed 31,055 women to prophylactic oxytocin during labor at five hospitals across Burkina Faso, Egypt, Turkey, and Vietnam. The 807 women in this group with uncontrolled blood-loss after delivery---a condition known as post-partum hemorrhage or PPH---were then randomized to receive either 800 milligrams misoprostol (treatment condition) or 40 international units oxytocin (control condition). There were 407 and 402 women in each treatment group, respectively. The primary outcome for the study was whether blood loss was controlled within 20 minutes after PPH diagnosis. The authors measure several important confounders like maternal age, blood loss at treatment, whether cord traction was maintained, maternal hemoglobin, whether the mother is currently married, whether the cord was clamped early, fetal gestational age, whether labor was augmented, whether labor was induced, maternal education, number of previous live births, whether the placenta was delivered prior to hemorrhage, and whether a uterine massage was given.

	\subsection{Modifications and methods}
	We modify the study in a couple of ways to make it more similar to an externally controlled trial. For each site, we separate the paired treatment groups and attempt to estimate effects using the units from other sites---\textit{e.g.}, for the misoprostol group in  Egypt we remove the oxytocin group from Egypt and attempt to estimate a causal effect using the oxytocin groups from the other sites. In this manner, we generate effect estimates and confidence intervals for each treatment group at each site.
	
	To estimate the treatment effects, we use COT with hyperparameter tuning as in Algorithm \ref{alg:bootwass} and the squared-Euclidean distance as the cost function.  We also compare COT to GLM, CBPS, SBW, SCM, and NNM. For GLM, CBPS, and SBW we utilize all first and second covariate moments. Finally, our estimator is that of Eq. \eqref{eq:hajek_init}.
	
	For our estimates to be valid, we require that there be no unmeasured confounding but also that the estimates are ``transportable,'' \textit{i.e.} we are able to take estimates from the other hospitals in the external group and ``transport'' them to site of interest. This requires that conditional on the observed covariates there are no other variables that can effect the outcome and treatment indicator ($d$-separation holds, \citealp{Pearl2013}).

	\subsection{Case study results}
	Amazingly, the three optimal transport flavored methods are able to achieve estimates close to the original treatment effects on average.
	In Figure \ref{fig:pph}, we can see that COT, NNM, and SCM do the best job in terms of average bias; moreover, these are the only methods that have estimates inside the original confidence interval. Of these, we note that COT has the least overall bias across all treatment groups and sites.
	
	In terms of inference, COT also performs best across all sites and treatment groups. Table \ref{tab:pph} displays how well the calculated confidence intervals cover the original treatment effects and also if the calculated estimates are inside the original confidence interval. In both cases, COT has the highest percentage of confidence intervals covering the original treatment effect and estimates inside the original confidence interval at $60\%$ each, respectively. And while NNM had a good overall bias, only 20\% of its confidence intervals covered the true treatment effect and its estimates were inside of only 20\% of the original confidence intervals.

	\begin{figure}[tb]
		\centering
		\includegraphics[width =\textwidth]{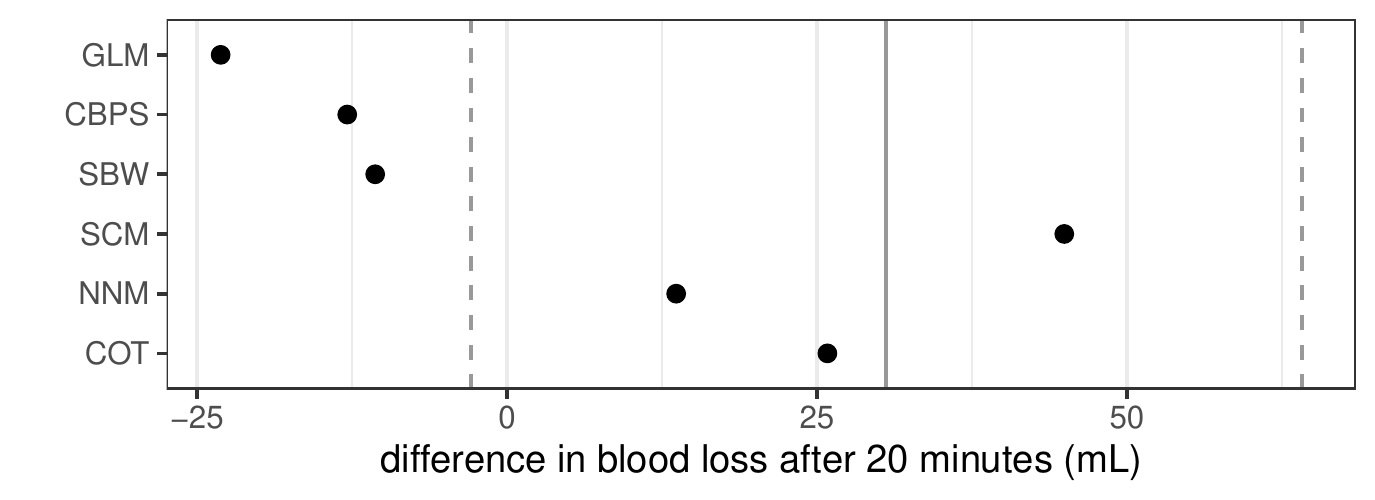}
		\caption{Results for treatment effect estimation averaged across treatment groups and study sites. The gray vertical line is the original treatment effect estimate for the entire study while the dotted vertical lines are the original confidence interval. The weighting methods under examination are logistic regression (GLM), Covariate Balancing Propensity Score (CBPS), Stable Balancing Weights (SBW), Synthetic Control Method (SCM), Nearest Neighbor Matching (NNM), and Causal Optimal Transport (COT).}
		\label{fig:pph}
	\end{figure}
	
	\input{tables/pph.tex}

	%
	
	
	\section{Summary and Remarks}\label{sec:summary}
	We have described a new tool for the estimation of causal effects in observational studies: Causal Optimal Transport. 
	This method allows for checks of distributional overlap and model-free weight estimation that is semiparametrically efficient.
	We also showed how several other methods for causal inference are closely related to COT and that COT can be seen as an interpolation between these methods. 
	In our simulation study, we demonstrated that this methodology performs well even when both the outcome and propensity score models are misspecified. Compared to other common weighting methods, COT generally has lower bias and lower root mean-squared error under model misspecification. 
	
	There are several areas for future research. 
	First, the sensitivity of COT to the choice of cost function remains to be elucidated.  
	Second, selecting covariates through typical model selection frameworks such as an $L_1$ penalized regression is not obvious given that COT does not generate clear predictive models, though this may not matter given the nonparametric nature of the weights. 
	Third, more work needs to be done to extend this framework to time series data, but further connections to SCM may offer a way forward.

	\if0\blind {
		\section*{Acknowledgments}
		The author would like to thank Claire Chaumont, Gang Liu, Aar\'on Sonabend, Lorenzo Trippa, and Jos\'e Zubizarreta for helpful comments and feedback on an earlier version of this manuscript.
		%
		This research was funded by generous support from NIH grant 5T32CA009337-40, the Department of Biostatistics at the Harvard T.H. Chan School of Public Health, and the David Geffen Scholarship from the David Geffen School of Medicine at UCLA. 
	} \fi
	
	\section*{Supplementary Materials}
	The Supplementary Materials contain proofs of the theorems in the paper, further empirical studies, an additional case study, and proofs for other formulations of COT. The last two s Proofs are found in Appendix \ref{sec:proofs}, and empirical studies of the convergence, confidence interval coverage, and efficacy of the tuning algorithm are located in Appendix \ref{sec:emp_studies}. Appendix \ref{sec:lalonde} is the additional case study and Appendix \ref{sec:other_ot} are the proofs for other formulations of COT.
	
	\printbibliography
	
	\newpage
	\appendix
	
	
	\section{Proofs} \label{sec:proofs}
	In this section we offer our proofs of the theorems and propositions stated in the paper.
	
	\subsection{Proof of Theorem \ref{thm:conv_cot}} \label{sec:conv_proof}
	In this section, we prove Theorem \ref{thm:conv_cot} which establishes the convergence of the COT weights. First, we need the following lemma.
	
	\begin{lemma}[The importance sampling weights converge to $\alpha$]\label{lemm:is_conv}
		Let Assumption \ref{assum:si} holds. Define the importance sampling weights as $\breve{w}_i^\star = \frac{d\alpha_z(X_i)}{d\alpha(X_i)}$ and define the self-normalized importance sampling weights as $w_i^\star =  \frac{1}{n}\breve{w}_i^\star/\sum_{i} \frac{1}{n}\breve{w}_i^\star.$ Then, 
		\[\boldw^\star \rightharpoonup \alpha.\]
	\end{lemma}
	\begin{proof}
		By Assumption \ref{assum:si}, $\breve{w}_i^\star$ exists for all $i$. Then define $\p_n (X \in E) = \sum_i \indicator(X_i \in E) w_i^\star$, for some $E \subset \mcX$.
		Take $f(X) = \indicator(X \in E)$. By Theorem 9.2 in \citet{Owen2013h}, $\p_n (X \in E) = \sum_i f(X_i) w_i^\star \as \E_\alpha(f(X)) = \int_\mcX f \, d \alpha = \p_\alpha(X \in E),$ where $\p_\alpha(\cdot)$ denotes the probability of $X \in E$ when $X \sim \alpha$. Thus we have $\lim_{n \to \infty} \p_n(X \in E) = \p_\alpha(X \in E)$ and the result follows.
	\end{proof}

	Now we are ready to proceed.
	\begin{proof}
		Under Assumption \ref{assum:si}, $\boldw^\star$ exists and by Lemma \ref{lemm:is_conv}, $\boldw^\star \rightharpoonup \alpha$.
		By Assumptions \ref{assum:conv} and \ref{assum:conv_rootn}, $S_\lambda$ is convex in its entries and metrizes the convergence in measure \citep{Feydy2018,janati_debiased_2020}. Thus,
		$0 \leq S_\lambda(\boldw_{\text{COT}}, \bolda) \leq S_\lambda(\boldc, \bolda) $ for all $\boldc \in \Delta_n$ that meet the chosen balancing constraints, $\delta$, which includes $\boldw^\star$ for a large enough $n$ and $m$. Hence,
		$0 \leq S_\lambda(\boldw_{\text{COT}}, \bolda) \leq S_\lambda(\boldw^\star, \bolda)$ and 
		since $\boldw^\star \rightharpoonup \alpha$,
		$S_\lambda(\boldw^\star, \bolda ) \to 0 \Rightarrow S_\lambda(\boldw_{\text{COT}}, \bolda) \to 0 \Rightarrow \boldw_{\text{COT}} \rightharpoonup \alpha $.
		
		We then take the same inequality and modify it slightly:
		$0 = S_\lambda(\bolda, \bolda) \leq S_\lambda(\boldw_{\text{COT}}, \bolda) \leq S_\lambda(\boldw^\star, \bolda).$
		If we add $\ot{\alpha}{\alpha}{\lambda} - \ot{\alpha}{\alpha}{\lambda}$ to each term and rearrange, we get
		\begin{align*}
			\frac{1}{2}\left\{\ot{\bolda}{\bolda}{\lambda} -  \ot{\alpha}{\alpha}{\lambda} \right\} &\leq 
			\ot{\boldw_{\text{COT}}}{\bolda}{\lambda} - \frac{1}{2} \ot{\boldw_{\text{COT}}}{\boldw_{\text{COT}}}{\lambda} - \frac{1}{2} \ot{\alpha}{\alpha}{\lambda} \\
			&\leq
			\ot{\boldw^\star}{\bolda}{\lambda} - \frac{1}{2} \ot{\boldw^\star}{\boldw^\star}{\lambda} - \frac{1}{2} \ot{\alpha}{\alpha}{\lambda}.
		\end{align*}
		For $n$  large enough, the terms in the last two inequalities will be approximately equal to 
		$\frac{1}{2}\{\ot{\boldw_{\text{COT}}}{\bolda}{\lambda} - \ot{\alpha}{\alpha}{\lambda} \}$ and $\frac{1}{2}\{\ot{\boldw^\star}{\bolda}{\lambda}-  \ot{\alpha}{\alpha}{\lambda} \},$ respectively.
		Therefore,
		\begin{align*}
			\frac{1}{2}\left\{\ot{\bolda}{\bolda}{\lambda} -  \ot{\alpha}{\alpha}{\lambda} \right\} &\leq 
			\frac{1}{2}\{\ot{\boldw_{\text{COT}}}{\bolda}{\lambda} - \ot{\alpha}{\alpha}{\lambda} \} \\
			&\leq
			\frac{1}{2}\{\ot{\boldw^\star}{\bolda}{\lambda} -  \ot{\alpha}{\alpha}{\lambda} \} 
		\end{align*}
		
		Under Assumptions \ref{assum:conv}--\ref{assum:conv_rootn}, Theorem 1 of \citet{Genevay2018} or Corollary 1 of \citet{Mena2019} hold and 
		$\E \{\ot{\bolda}{\bolda}{\lambda} -  \ot{\alpha}{\alpha}{\lambda} \} = \mcO\left( \frac{1}{\sqrt{n}}\right) $ and
		$\E \{\ot{\boldw^\star}{\bolda}{\lambda} -  \ot{\alpha}{\alpha}{\lambda}\}  = \mcO\left( \frac{1}{\sqrt{n}}\right).$ Thus,
		$\E \{\ot{\boldw_{\text{COT}}}{\bolda}{\lambda} - \ot{\alpha}{\alpha}{\lambda}\} = \mcO\left( \frac{1}{\sqrt{n}}\right)$.

	\end{proof}
	
	\subsection{Proof of Theorem \ref{thm:vopt}}\label{sec:vopt_proof}
	\begin{proof}
		Define $e_{z,i} = \p(Z_i = z \given X_i, S_i = 1) \p(S_i = 1 \given X_i)/\p(S_i = 0 \given X_i)$, or the inverse weight targeting the TATE. This is also the inverse of the Radon-Nikodym derivative. As a reminder, $\mu_z(X) = \E\{Y(z) \given X\}$. 
		We first decompose $\hat{\tau} - \tau$ into several residual terms:
		\begin{align*}
			\hat{\tau} - \tau &= \sum_{i} w_i Z_i Y_i -  \sum_{i} w_i (1-Z_i) Y_i - \tau\\
			&= \sum_{i} w_i Z_i \{Y_i - \mu_1(X_i)\} -  \sum_{i} w_i (1-Z_i) \{Y_i - \mu_0(X_i)\} \\
			&\;\;\;\;\;\;\;\; + \sum_{i} w_i Z_i \mu_1(X_i) -  \sum_{i} w_i (1-Z_i) \mu_0(X_i)- \tau \\
			& = \sum_{i} w_i Z_i \{Y_i - \mu_1(X_i)\} -  \sum_{i} w_i (1-Z_i) \{Y_i - \mu_0(X_i)\} \\
			&\;\;\;\;\;\;\;\;+ \sum_{i} w_i Z_i \mu_1(X_i) -  \sum_{i} w_i (1-Z_i) \mu_0(X_i)- \tau \\
			&= \frac{1}{n}\sum_{i}  \frac{Z_i}{e_{1,i}} \{Y_i - \mu_1(X_i)\} -  \frac{1}{n}\sum_{i} \frac{1-Z_i}{e_{0,i}} \{Y_i - \mu_0(X_i)\} \\ 
			&\;\;\;\;\;\;\;\; + \sum_{i} \left(w_i - \frac{1}{n \cdot e_{1,i}}\right) Z_i \{Y_i - \mu_1(X_i)\} \\
			&\;\;\;\;\;\;\;\; -  \sum_{i} \left(w_i - \frac{1}{n \cdot e_{0,i}}\right) (1-Z_i) \{Y_i - \mu_0(X_i)\} \\ 
			&\;\;\;\;\;\;\;\;+ \sum_{i} w_i Z_i \mu_1(X_i) -  \sum_{i} w_i (1-Z_i) \mu_0(X_i)\\
			&\;\;\;\;\;\;\;\; - \frac{1}{m}\left\{\sum_{j} \mu_1(X_j) -  \sum_{j} \mu_0(X_j) \right\}\\
			&\;\;\;\;\;\;\;\;+ \frac{1}{m}\left\{\sum_{j} \mu_1(X_j) -  \sum_{j} \mu_0(X_j) \right\} - \tau \\
			&= A + B + C,
		\end{align*}
		where 
		\begin{align*}
			A &= \frac{1}{n}\sum_{i}  \frac{Z_i}{e_{1,i}} \{Y_i - \mu_1(X_i)\} -  \frac{1}{n}\sum_{i} \frac{1-Z_i}{e_{0,i}} \{Y_i - \mu_0(X_i)\} \\ 
			&\;\;\;\;\;\;\;\;+ \frac{1}{m}\left\{\sum_{j} \mu_1(X_j) -  \sum_{j} \mu_0(X_j) \right\} - \tau\\
			B &=  \frac{1}{n} \sum_{i} \left(n \cdot w_i - \frac{1}{ e_{1,i}}\right) Z_i \{Y_i - \mu_1(X_i)\} \\
			&\;\;\;\;\;\;\;\; -   \frac{1}{n} \sum_{i} \left(n \cdot w_i - \frac{1}{  e_{0,i}}\right) (1-Z_i) \{Y_i - \mu_0(X_i)\}\\
			C & = \sum_{i} w_i Z_i \mu_1(X_i) -  \sum_{i} w_i (1-Z_i) \mu_0(X_i)\\
			&\;\;\;\;\;\;\;\; - \frac{1}{m}\left\{\sum_{j} \mu_1(X_j) -  \sum_{j} \mu_0(X_j) \right\}.
		\end{align*}
		The goal is to show that both $n^{1/2} B$ and $n^{1/2}C$ are  $o_p(1)$. Then, since $A$ has the form of the semiparametrically efficient score function, the result follows.
		
		First, for $B$, we have that $\lim_{n \to \infty} \boldw = \lim_{n \to \infty} \boldw^\star$ by Corollary \ref{coro:cot_conv_is}. This also implies that $\lim_{n \to \infty} n \cdot \boldw_i = 1/e_i$ since by the Radon-Nikodym theorem, the Radon-Nikodym derivative is unique almost surely. To prove that $n^{-1/2} B$ goes to 0, it will be sufficient to prove that $\frac{\sqrt{n}}{n} \sum_{i} \left | \left(n \cdot w_i - \frac{1}{ e_{z,i}}\right) \indicator(Z_i = z) \left\{Y_i - \mu_z(X_i)\right\}  \right | \to 0$. We then have
		\begin{align*}
			\frac{\sqrt{n}}{n} \sum_{i} & \left | \left(n \cdot w_i - \frac{1}{ e_{z,i}}\right) \indicator(Z_i = z) \left\{Y_i - \mu_z(X_i)\right\}  \right | \\
			&\leq \frac{\sqrt{n}}{n} \sum_{i} \left | \left(n \cdot w_i - \frac{1}{ e_{z,i}}\right) \indicator(Z_i = z) \right |  \left |Y_i - \mu_z(X_i)  \right |\\
			&\leq  \operatorname{ess.} \sup_i \left | \left(n \cdot w_i - \frac{1}{ e_{z,i}}\right) \indicator(Z_i = z) \right | \frac{1}{\sqrt{n}} \sum_{i}   \left |Y_i - \mu_z(X_i)  \right |,
		\end{align*}
		where $\operatorname{ess.} \sup$ is the essential supremum.
		The essential supremum quantity goes to 0 as a consequence of Corollary \ref{coro:cot_conv_is} while the residual quantity has finite expectation and variance by assumption. This implies $\frac{1}{\sqrt{n}} \sum_{i}   \left |Y_i - \mu_z(X_i)  \right | \dist L$ for some random variable $L$ and that the desired quantity goes to 0 by Slutsky's theorem.
		
		Second, for $C$, we have by assumption that $S_\lambda(\boldw, \bolda) = o_p(n^{-1/2})$. 
		This implies that the empirical expectations also converge at a faster than $\sqrt{n}$-rate. Alternatively, for the basis function constraints, we have \[\left| \E_\boldw\{\indicator(Z = z)\mu_z(X)\} - \E_\bolda\{\mu_z(X)\} \right| \leq \sum_k \delta_k |\gamma_k| \leq \| \delta \|_2^2 \|\gamma\|_2^2.\] Thus, for each value $z$ of $Z$,
		$ \sqrt{n} \left| \E_\boldw\{\indicator(Z = z)\mu_z(X)\} - \E_\bolda\{\mu_z(X)\} \right| = o_p(1).$
		
		Finally, with Assumption \ref{assum:sutva}, we can replace $Y_i$ with $Y_i(Z_i)$. This means that by Assumption \ref{assum:sutva} and Assumption \ref{assum:outcome_reg}, $\E(A) = 0$ and $\lim_{n \to \infty}\E(\hat{\tau}) = \tau$. Thus, under Assumptions \ref{assum:sutva}--\ref{assum:outcome_reg} and by the fact that $A$ has the form of the semiparametrically efficient score function, $A$ converges to $\tau$ at a $\sqrt{n}$-rate and $\sqrt{n}(\hat{\tau} - \tau)$ has the desired asymptotic distribution. 
	\end{proof}

	\subsection{Proof of Theorem \ref{thm:bp_conv}} \label{sec:bp_proof}
	Before we proceed to our proof, a slight digression is necessary to discuss how the barycentric projection in the setting of this paper will differ slightly from the usual formulation in optimal transport. Typically in optimal transport problems, we would include all of the available data in our distance metric; however, the missing potential outcomes make this inadvisable. Simply throwing in the observed outcomes into the optimal transport problem  could lead to weights that bias treatment effects towards zero. Thus, we estimate the optimal transport plan only on the covariate data $X$ and then incorporate the outcomes after estimation of the transport plan.
	
	Further, we require the following definition.
	\begin{definition}
		The primal form of \eqref{eq:ot_reg} with an entropy penalty for generic measures $\alpha$ and $\beta$ on $\mcX$ is equivalent to 
		\begin{equation}
			\label{eq:ot_kl}
			\ot{\alpha}{\beta}{\lambda} = \inf_{\pi \in \boldU(\alpha, \beta)} \int_{\mcX \times \mcX} c(x,y)d\pi(x,y) + \lambda  \log\left(\frac{\pi(x,y)}{\alpha(x) \otimes \beta(y) } \right)d\pi(x,y).
		\end{equation}
		Eq. \eqref{eq:ot_kl} has the dual form
		\begin{equation}
			\label{eq:ot_dual}
			\sup_{f,g} \int_\mcX f(x) d\alpha(x) + \int_\mcX g(y) d \beta(y) - \lambda \int_{\mcX \times \mcX} e^{(f(x) + g(y) - c(x,y))/\lambda}d\alpha(x)d\beta(y) + \lambda,
		\end{equation}
		with a primal solution equal to $ e^{(f(x) + g(y) - c(x,y))/\lambda}d\alpha(x)d\beta(y)$.
		\label{def:dual_and_such}
	\end{definition}
	\noindent We note that the dual form is justified by Fenchel-Rockafellar duality but defer a proof to sources such as \citet{Peyre2019}. 
	
	\begin{proof}
		Denote $\lim_{n \to \infty} \boldw = \omega$ and $\lim_{n \to \infty} \bolda = \alpha$. By assumption, $\omega = \alpha$. Without loss of essential generality, assume the weights only adjust one treatment group towards the full sample. We also assume both groups have equal sample sizes since it will make some of the notation easier to follow.
		
		We will denote the optimal transportation plan between $\omega$ and $\alpha$ as $\pi$. The barycentric projection from $\alpha$ into $\omega$ is then $\int_\mcX x d \pi(x \given x')$, where $d\pi(x \given x') = d\pi(x, x')/d\alpha(x')$. 
		Under the assumptions of the theorem, this transport plan, $\pi$ is unique and is supported on the graph of a Monge map \citep{brenier_decomposition_1987}. Moreover, the barycentric projection  will be the optimal map, $T$ \citep[Lemma 12.2.3]{Ambrosio2005}. Finally, the optimal transport plans in finite samples will converge to the limiting value: $\boldP_n \to \pi$ \citep[Theorem 5.20]{villani_cedric_optimal_2008}. We also note that the Monge map $T$ will be the identity function since the distributions are the same.
		
		\paragraph{$\mu_z$ is like a Monge map for the outcomes.} 	
		For ease of exposition, we can think of units from $\omega$ as having their outcome $Y(z)$ observed, while the outcomes are completely missing for units in $\alpha$.
		Thus, we need some way of projecting an individual from $\alpha$ to $\omega$ and generating their hypothetical outcome.
		For an individual with covariate values $x'$, this  will be
		\[\int \begin{bsmallmatrix} y(z) \\ x \end{bsmallmatrix} dp(y(z)\given x) d \pi(x \given x') = \int \begin{bsmallmatrix} y \\ x \end{bsmallmatrix} dp(y\given x) d \pi(x \given x') =\int \begin{bsmallmatrix} \mu_z(x) \\ x \end{bsmallmatrix} d\pi(x \given x') =   \begin{bsmallmatrix} \mu_z(x') \\ x' \end{bsmallmatrix}, \] where
		under Assumptions \ref{assum:sutva}--\ref{assum:si}, the observed outcomes can be used for the potential outcomes and there is no interference.
		
		This means the Monge map is $\begin{bsmallmatrix} \mu_z(x') \\ x' \end{bsmallmatrix}$ in our particular case when we map from a space without outcomes to a space with outcomes. The failure of the values of $x'$ to change aligns with our assumption that the covariates precede treatment causally so their value does not change depending on the treatment group. Also, we note that since the outcomes are unknown, the best projection we can hope for is the conditional mean outcome function.
		
		Given this fact and that we do not care about projecting the covariates, we will focus on the barycentric projection of just the outcome:
		\[\bar{\pi}(x') = \int y dp(y \given x) d\pi(x \given x') =  \int \mu_z(x) d\pi(x \given x')= \mu_z(x') = T(x'). \]
		
		\paragraph{The estimated barycentric projection converges to $\mu_z$.} 
		Denote $\bar{\boldP}^\lambda$ as the distribution implied by the barycentric projection from the entropy-penalized optimal transport problem and $\bar{\pi} = \pi(\cdot \given x')$ as the optimal barycentric projection distribution.
		First,  we rewrite our equation to make the barycentric projections explicit. 
		Starting with the $L_2$ metric, we can represent $Y$ as the convolution of its conditional mean, $\mu_z$, with another random variable with mean zero, $\epsilon$:
		\begin{align*}
			\left\| \int y \dpl -  \int\mu_z \dpi \right\| &= \left \| \int (\mu_z + \epsilon) \dpl -  \int\mu_z \dpi  \right\| \\
			&\leq \left \| \int \mu_z (\dpl - \dpi) \right\| + \left \|\int \epsilon \dpl \right\| \\
			& \leq L \left \| \int x (\dpl - \dpi) \right\| + \left \|\int \epsilon \dpl \right\|.
		\end{align*}
		Then taking the square and expectation:
		\begin{align*}
			\E_\alpha \left\| \int y \dpl -  \int\mu_z \dpi \right\|^2 &\leq L^2 \E_\alpha \left \| \int x (\dpl - \dpi) \right\|^2 \\
			& \qquad + 2 L \E_\alpha\left \| \int x (\dpl - \dpi) \right\|\left \|\int \epsilon \dpl \right\| \\
			& \qquad +  \E_\alpha \left \|\int \epsilon \dpl \right\|^2\\
			&\leq L^2 \underbrace{\E_\alpha\left \| \int \mu_z (\dpl - \dpi) \right\|^2}_\text{E} \\
			& \qquad + 2L\underbrace{\sqrt{\E_\alpha\left \| \int \mu_z (\dpl - \dpi) \right\|^2} \sqrt{\E_\alpha\left \|\int \epsilon \dpl \right\|^2}}_\text{F} \\
			&\qquad +  \underbrace{\E_\alpha \left \|\int \epsilon \dpl \right\|^2}_\text{G}
		\end{align*}
		
		For the term in $G$, we have
		\begin{align*}
			G &= \E_\alpha \left(\sum_i \epsilon_i \overline{\boldP}_i^\lambda \right)^2\\
			&= \E_\alpha \left(\sum_i \sum_{i'} \epsilon_i \epsilon_{i'} \overline{\boldP}_i^\lambda \overline{\boldP}_{i'}^\lambda \right)\\
			&= \E_\alpha \E\left(\sum_i \sum_{i'} \epsilon_i \epsilon_{i'} \overline{\boldP}_i^\lambda \overline{\boldP}_{i'}^\lambda \given X,X'\right)\\
			&\leq \E_\alpha \left(\sum_i  \xi^2 (\overline{\boldP}_i^\lambda)^2\right).
		\end{align*}
		We then argue that this term goes to zero since $\overline{\boldP}_i^\lambda$ converges to $\overline{\pi}^\lambda$, which is without an atom.
		By the primal solution in Definition \ref{def:dual_and_such} and the structure of the penalty in of Eq. \eqref{eq:ot_reg}, we know $\overline{\boldP}_i^\lambda = e^{[f(x_i) + g(x') - c(x,x')]/\lambda}w_i < 1, \forall i$.
		Further, because $(\overline{\boldP}_i^\lambda)^2 < \overline{\boldP}_i^\lambda$, each value of the square is bounded away from one as well. 
		%
		%
		Additionally, $\overline{\boldP}_i^\lambda$ will converge to zero for all $i$ since $\pi^\lambda$ has a density with respect to $\omega$ and $\alpha$ \citep{Peyre2019}. This also implies that the empirical mean and variance of $\overline{\boldP}_i^\lambda$ will also go to zero and, thus, $\sum_i (\overline{\boldP}_i^\lambda)^2 \to 0$. 
		
		The rate of this convergence will be determined by the rate of convergence of the regularized optimal transport problem, \[K_d \lambda \left(1 + \frac{\sigma^{\lceil 5d/2\rceil + 6}}{\lambda^{d' + 3}}\right) n^{-1/2}\] for a constant $K_d$ depending only on the dimension $d$, $\sigma$ a constant determined by the subgaussian tail of the random variable, and $d' = \lceil 5d/4 \rceil$  \citep[Corollary 1]{Mena2019}. As long as $\lambda$ goes to zero slowly enough, the sum will converge to zero. This will be the case if we take 
		\begin{equation}
			\lambda \asymp n^{- \frac{1}{2 d' + 9}}.
			\label{eq:lambda_n_val}
		\end{equation}
		Lastly, 
		\[\E_\alpha \left(\sum_i  \xi^2 (\overline{\boldP}_i^\lambda)^2\right) \to 0\]
		by Lebesgue's Dominated Convergence Theorem.

		Turning next to the term in $E$,
		we have by Theorems 4 and 5 in \citet{pooladian_entropic_2021} that 
		\[E \lesssim \lambda^{1-d''/2} \log(n) n^{-1/2} + \lambda ^{(t +1)/2} + \lambda^2 I_0 + \lambda \inv(1 + \lambda^{1-d''/2})\log(n) n^{-1/2} ,\]
		where $I_0$ is the Fisher Information of $\alpha$, $d'' = 2 \lceil d/2 \rceil$, and $t \in [2, 3]$ is related to the number of derivatives that exist for the optimal Kantorovich potential \citep{Chizat2020, pooladian_entropic_2021}. 
		These theorems also require several assumptions which hold for our particular case. First, in our particular case where Brenier's Theorem holds, the derivatives of the dual potentials of the unpenalized problem equal the Monge map. These maps are then infinitely differentiable since in our case they are equal to the identity function. This also means $t =3$. Also since we assume $I_0$ exists, we do not need the bounded densities required in \citet{pooladian_entropic_2021}.
		Then using the value of $\lambda$ from \eqref{eq:lambda_n_val}, we get
		\[E \lesssim (1 + I_0) n^{- \frac{2}{2 d' + 9 }}.\]
		
		Finally, to establish the rate is determined by $E$, we need for the term in $F$ to go to zero faster than the term in $E$. To see this, we have that
		\[\frac{2LF}{E} = \frac{2L\sqrt{E}\sqrt{G}}{E} = \frac{2L\sqrt{G}}{\sqrt{E}}.\] 
		Since the numerator goes to zero faster than the denominator, the overall rate is determined by $E$.

	\end{proof}
	
	\subsection{Proof of Proposition \ref{prop:l2_map}}
	\begin{proof}
		It will be sufficient to show the proof for the estimator of $\E Y(z)$ for one such group $Z=z$. The barycentric projection estimator of the mean in this group will be $n \inv \sum_i \hat{Y}_i(z)$. Expanding the term,
		\begin{align*}
			n \inv \sum_i \hat{Y}_i(z) &= n\inv \sum_i n \sum_j Y_j \indicator(Z_j = z) \boldP_{ji}\\
			&= \sum_j Y_j \indicator(Z_j = z) \sum_i \boldP_{ji}\\
			&= \sum_j Y_j \indicator(Z_j = z) w_j,
		\end{align*}
		where the first equality is by definition and the last equality is by the constraints of the optimal transport problem.
	\end{proof}
	
	\subsection{Proof of Proposition \ref{prop:scm}} \label{sec:scm_proof}
	\begin{proof}
		
		Under the assumptions of the proposition, the limiting optimal transport plan is unique and supported on the graph of a Monge map, the Monge map will be equal to the barycentric projection, and the empirical transport plans converge to the limiting value \citep{brenier_decomposition_1987, Ambrosio2005, villani_cedric_optimal_2008}.
		
		Let $\hat{T}(X_j) = \sum_{i: Z_i = z} X_i w_i$. We also take our data to be univariate for notational simplicity but the results are easily extended to data in $\R^d$. Then the objective for SCM can be written as
		\begin{align*}
			\operatorname{SCM }& = \inf_{\hat{T}} n \inv \sum_j \left \|\hat{T}(X_j) - X_j \right \|_2^2 \\
			&=  n\inv \sum_j \hat{T}(X_j)^2 - 2 \hat{T}(X_j) X_j + X_j^2.
		\end{align*}
		We can represent $\hat{T}(X_j)$  as a matrix of weights with certain constraints instead of the sum over the weights for each $j$:
		\begin{align} 
			\operatorname{SCM} & = \inf_{\tilde{\boldP} \in \R^{n_z \times n}_+: \tilde{\boldP} \trans \ones = \ones}
			n\inv \sum_j \left(\sum_{i: Z_i = z} X_i \tilde{\boldP}_{ij}\right)^2 - 2 X_j \sum_{i: Z_i = z} X_i \tilde{\boldP}_{ij}  + X_j^2 \nonumber \\
			&= \inf_{\tilde{\boldP} \in \R^{n_z \times n}_+: \tilde{\boldP} \trans \ones = \ones}
			n\inv \sum_j \sum_{i: Z_i = z} X_i \tilde{\boldP}_{ij} \tilde{X}_j - 2 \sum_{i: Z_i = z}X_j  X_i \tilde{\boldP}_{ij}   + X_j^2 \nonumber \\
			&= \inf_{\boldP \in \R^{n_z \times n}_+: \boldP \trans \ones = \bolda}  \sum_{i: Z_i = z, j} X_i   \tilde{X}_j \boldP_{ij} - 2  \sum_{i: Z_i = z,j }X_j  X_i\boldP_{ij}   +  \sum_{i: Z_i = z,j } X_j^2\boldP_{ij} \label{eq:scm_final}
		\end{align}
		The terms in $ \tilde{X}_j$ are equivalent to a barycentric projection from $\bolda$ into $\bolda_z$. And we can see that the solution $\boldP_{ij}$ is equivalent to an optimal transport plan with one of the margins allowed to vary. The corresponding problem under a Kantorovich relaxation is
		\begin{align}
			\inf_{\boldP \in \R^{n_z \times n}_+: \boldP \trans \ones = \bolda}  \sum_{i: Z_i = z, j} X_i ^2 \boldP_{ij} - 2  \sum_{i: Z_i = z,j }X_j  X_i\boldP_{ij}   +  \sum_{i: Z_i = z,j } X_j^2\boldP_{ij}.
			\label{eq:kant_final}
		\end{align}
		The marginal distributions $\boldw$ from both of these problems can be found as $\boldw = \boldP \ones$.
		Eqs. \eqref{eq:scm_final} and \eqref{eq:kant_final} will be equal in finite samples if $\sum_j \tilde{X}_j \boldP_{ij} = X_i \sum_{j} \boldP_{ij}$, which is not guaranteed. However, in the limit they will be the same, as we show next.
		
		Under a similar argument to the proof of Theorem \ref{thm:conv_cot},
		$\ot{\boldw_{NNM}}{\bolda}{} \leq \ot{\boldw^\star}{\bolda}{} \to 0$. Then, denoting the Monge problem using barycentric projections as $\operatorname{M}$,
		$\operatorname{M}(\boldw_{SCM}, \bolda)  \leq \operatorname{M}(\boldw_{NNM}, \bolda) \to 0$. The last term holds since $
		\boldw_{NNM} \to \alpha$ and $\int x d\pi(x \given x') = x'$.
		
		Finally, to see the equivalence of the Monge and Kantorovich formulations for fixed margins $\alpha ,\; \omega$ under the given assumptions:
		\begin{align*}
			\inf_T \int | T(x') - x' |^2 d \alpha &= \inf_T  \int T(x')^2 d\alpha + \int (x')^2 d\alpha - 2 \int x' T(x') d\alpha \\
			&= \inf_T  \int x^2 dT_\# \alpha + \int (x')^2 d\alpha - 2 \int x T(x) d \alpha \\
			&= \int x^2 d\omega + \int (x')^2 d\alpha  - \inf_\pi \int x' \left[ \int x d\pi(x \given x')\right] d \alpha\\
			&= \int x^2 d\omega + \int (x')^2 d\alpha  - \inf_\pi \int  x x'  d\pi(x , x')\\
			&= \inf_{\pi \in \boldU (\omega, \alpha)} \int |x - x'|^2 d \pi(x, x')
		\end{align*}
	\end{proof}
	
	\section{Further empirical studies} \label{sec:emp_studies}

	\subsection{Empirical convergence} \label{sec:emp_conv}
	
	The data 
	generating model in this section comes from the setting in Section \ref{sec:sims} with high-overlap between covariate distributions and the estimand of interest is the ATE. 
	We examine three measures of performance of the estimated weights, $\boldw$, at approximating 1) the target empirical distribution, $\bolda$, in terms of 2-Sinkhorn divergence, 2) the distribution of the self-normalized propensity score, $\boldw^\star$, in terms of the 2-Sinkhorn divergence, and 3) the difference between $\boldw$ and $\boldw^\star$ under an $L_2$ norm. We present averages across 1000 replications.
	
	The comparator methods we use are: a Probit Generalized Linear Model (GLM), the true data-generating model; Stable Balancing Weights (SBW) using the correct propensity score covariate functions; and Nearest Neighbor Matching (NNM) with replacement, which is of course equivalent to an unpenalized COT (see Section \ref{sec:compare}). For NNM, we use a cost function that is equal to $\| \cdot \|_p^p$ with $p = d/2+1$ to meet the conditions of Theorem 1 in \citet{Fournier2015}. 
	
	In Figure \ref{fig:wass_conv_B}, we see that the COT weights do a better job of approximating the target distribution under the 2-Sinkhorn Divergence but that the GLM model does better at targeting the distribution implied by the true inverse propensity score under the same metric. Of note, SBW displays decaying rates convergence as the sample size increases and even performs worse than NNM for large sample sizes.
	
	Figure \ref{fig:l2_conv_B} displays the convergence in $L_2$-norm for the various methods. As we would expect, the GLM model converges fastest to the values of the true inverse propensity score. The COT weights using the Sinkhorn divergence display slightly worse rates of convergence, on average, followed by NNM. SBW again displays a rate that decays with the sample size, though it does perform better than other methods when sample sizes are small.
	\begin{figure}[!tb]
		\includegraphics[width=\textwidth]{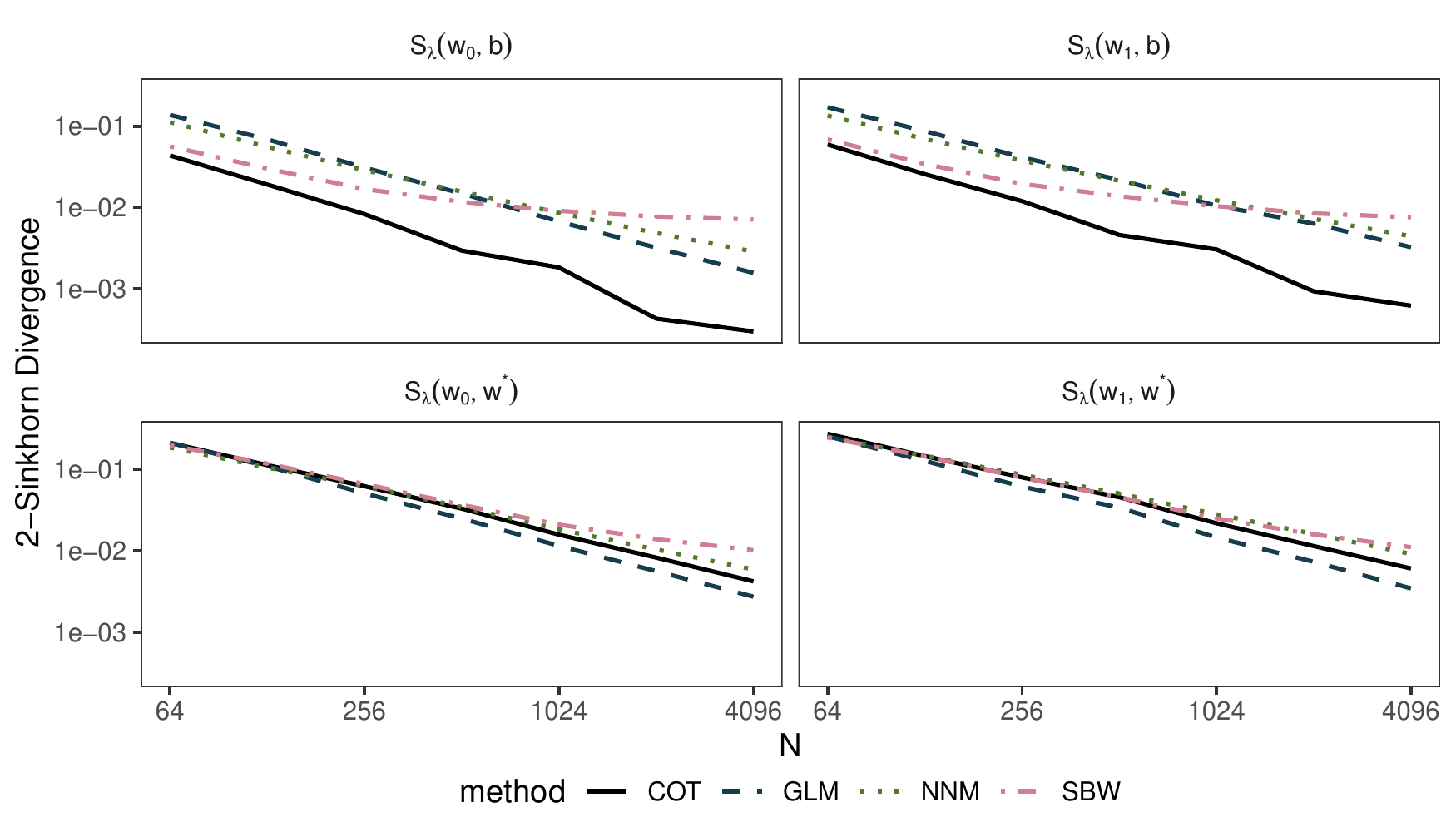}
		\caption{Convergence of the weights to the distributions specified by the empirical distributions (top) and the distributions specified by the true propensity score/Radon-Nikodym derivatives (bottom). Weights are a Causal Optimal Transport (COT), Nearest Neighbor Matching (NNM), a Probit model (GLM), and Stable Balancing Weights (SBW). Lines denote means across 1000 simulations. Both axes are on the log scale.}
		\label{fig:wass_conv_B}
	\end{figure}
	\begin{figure}[!tb]
		\includegraphics[width=\textwidth]{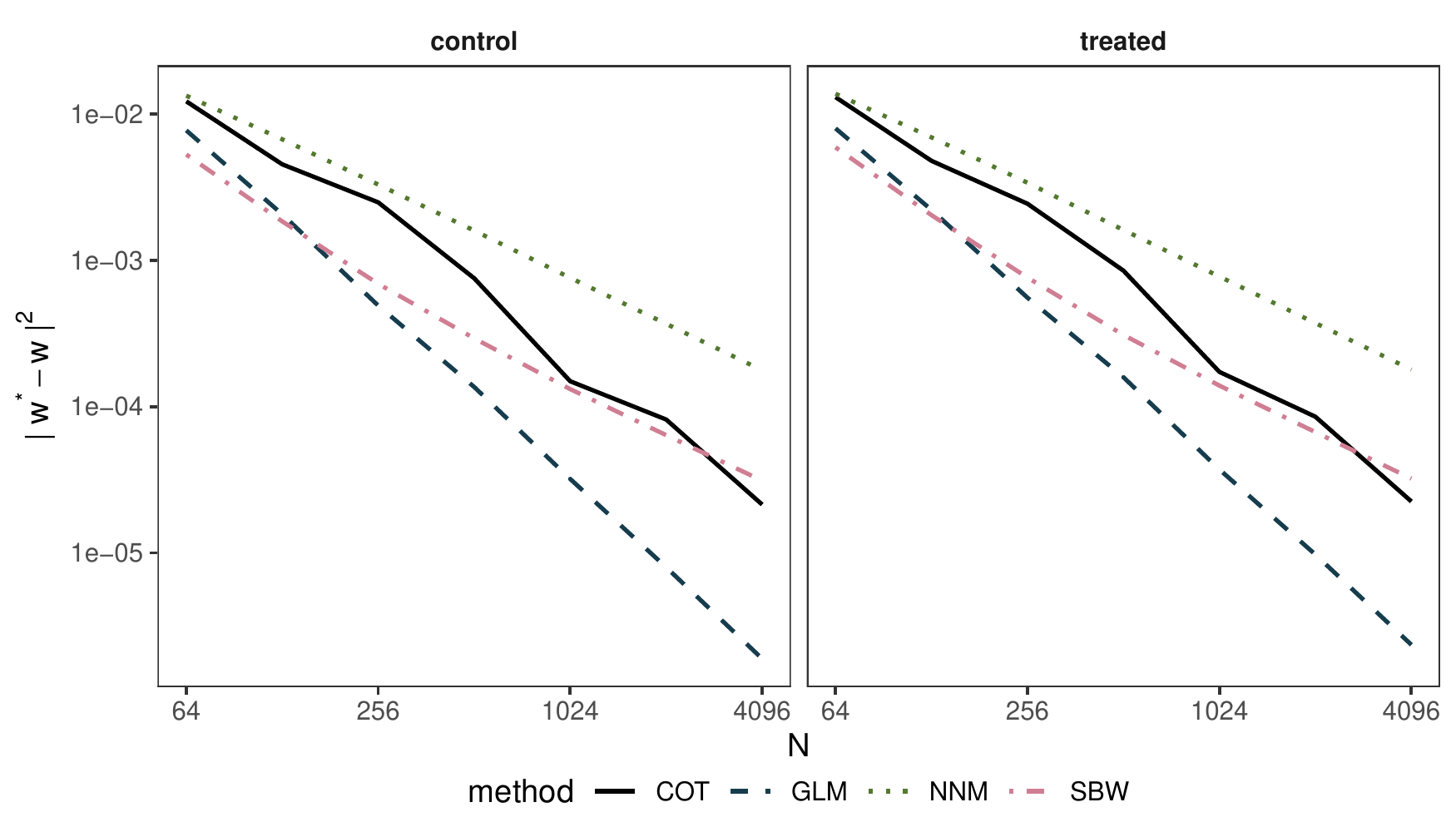}
		\caption{Convergence of the estimated weights to the values of the true inverse propensity score in terms of the $L_2$ norm. Weights are a Causal Optimal Transport (COT), Nearest Neighbor Matching (NNM), a Probit model (GLM), and Stable Balancing Weights (SBW). Lines denote means across 1000 simulations. Both axes are on the log scale.}
		\label{fig:l2_conv_B}
	\end{figure}

	\subsection{Empirical coverage of asymptotic confidence interval} \label{sec:vopt_cov}
	The empirical coverage of the confidence interval is the focus of this subsection. We utilize the generating model from the setting in Section \ref{sec:sims} with high-overlap between covariate distributions and a linear outcome model linear outcome model 
	\[Y(0) = Y(1) = X_1 + X_2 + X_3 - X_4 + X_5 + X_6 + \eta \]
	with $\eta \sim \N(0,1)$. 
	The target distribution, $\bolda$, is the full sample, making the estimand of interest the ATE. We run 1000 replications of our experiment. 
	\begin{figure}[!tb]
		\centering
		\begin{subfigure}[t]{\linewidth}
			\includegraphics[width=\linewidth]{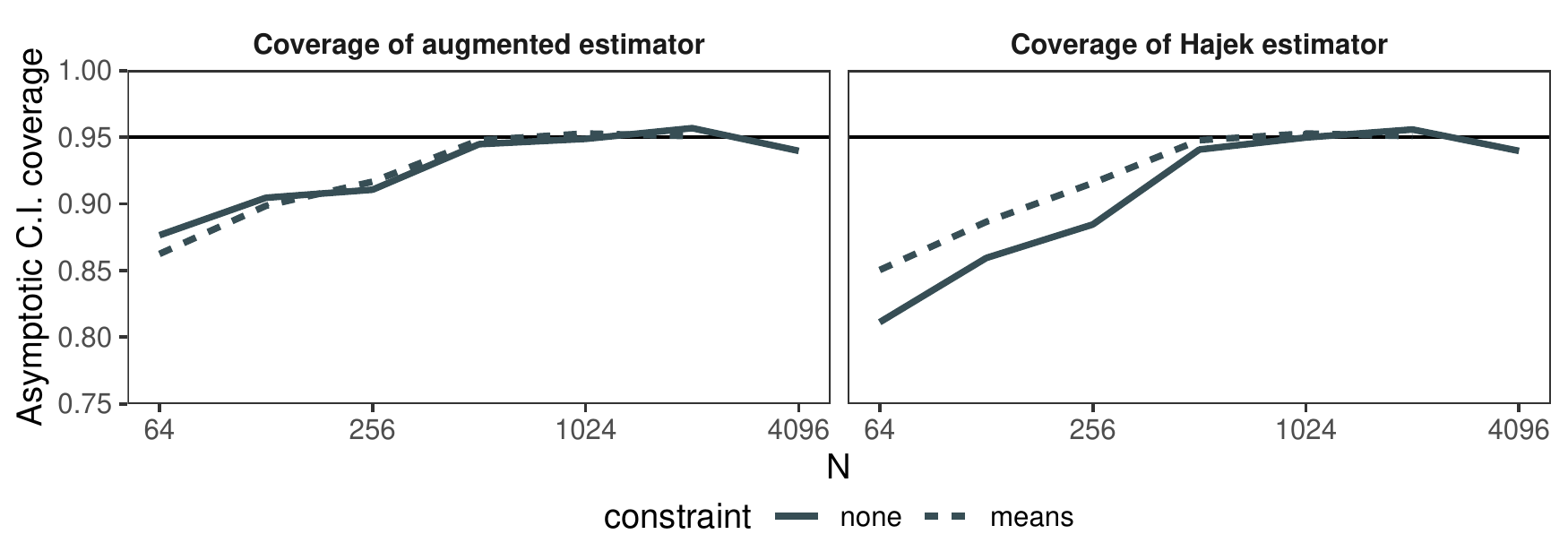}
			\caption{Coverage of the true treatment effect}
			\label{fig:ot_ci_cov_true} 
		\end{subfigure}
		\begin{subfigure}[b]{\linewidth}
			\includegraphics[width=\linewidth]{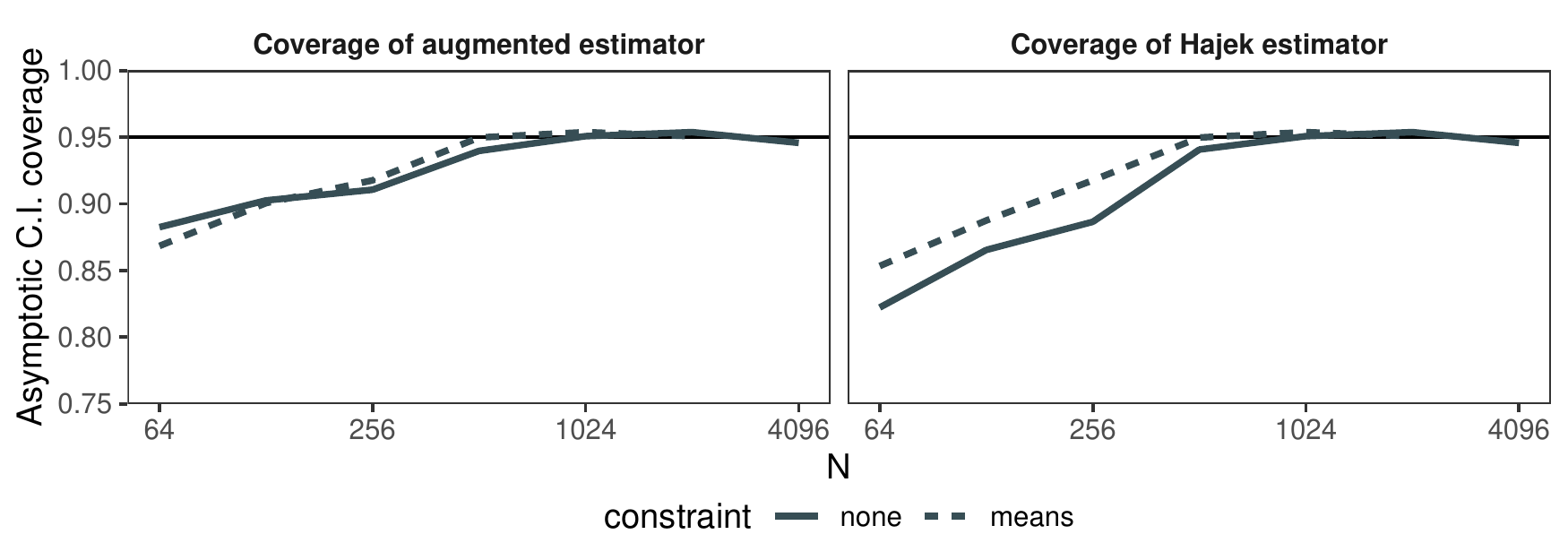}
			\caption{Coverage of the estimated average treatment effect}
			\label{fig:ot_ci_cov_expect}
		\end{subfigure}
		\caption{Coverage of the asymptotic confidence interval for both the true effect (a) and the average estimated effect (b). Solid lines denote no basis function balancing (``none'') and the balancing of the means of covariates (``means'').}
		\label{fig:ot_ci_cov}
	\end{figure}
	
	Figure \ref{fig:ot_ci_cov} displays results for increasing sample sizes. In the top part, Figure \ref{fig:ot_ci_cov_true}, we examine the coverage of the true estimate of zero in a variety of settings. Amazingly, the COT method achieves well-calibrated confidence intervals without using an augmented estimator or mean constraints. Similar results are observed for the empirical expectations in Figure \ref{fig:ot_ci_cov_expect}. In both cases, the non-augmented balancing constraint method converges a bit faster than the non-augmented method without balancing constraints.
	
	\subsection{Tuning algorithm} \label{sec:tun_alg}
	In this section, we examine the performance of the tuning algorithm presented in Algorithm \ref{alg:bootwass}. We again use the setting of Section \ref{sec:sims} with high-overlap between covariate distributions, use COT with an $L_2$ metric and no balancing functions, and for a variety sample sizes from 32 to 4096. The target distribution in this case is the full sample making the estimand the ATE. Performance is measured in terms of an Anderson-Darling statistic between the estimated weights, $\boldw$, and the self-normalized inverse propensity score, $\boldw^\star$:
	\[\frac{1}{n}\sum_i^n \frac{(w_i - w_i^\star)^2}{w_i^\star(1 - w_i^\star)}.\]
	We use this term rather than a simple $L_2$ norm because it will appropriately adjust for the discrepancy between weight vectors as the values become small. Finally, we run this experiment 1000 times.
	
	Figure \ref{fig:tune} displays the results for the tuning algorithm as the sample size increases. We can see that initially the algorithm avoids the highest discrepancy area for intermediate values of the penalty parameter $\lambda$. As the sample size increases, the algorithm concentrates on $\lambda$ values that minimize the difference between the true and estimated propensity scores. This holds true for both the treated and control observations.

	\begin{figure}[!tb]
		\centering
		\begin{subfigure}[t]{\linewidth}
			\includegraphics[width=\linewidth]{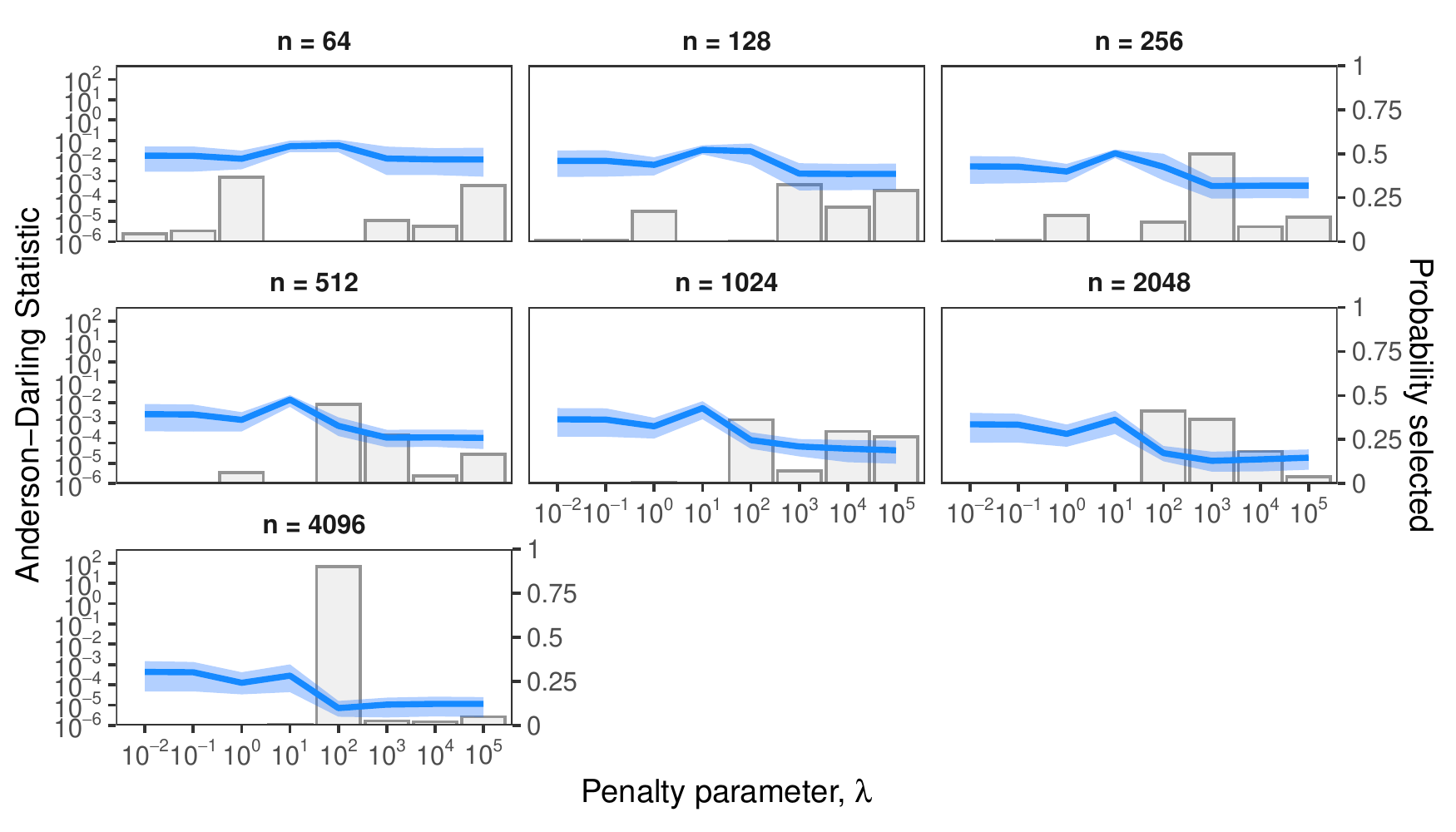}
			\caption{Selection of penalty parameter, $\lambda$, for the treated}
		\end{subfigure}
		\begin{subfigure}[b]{\linewidth}
			\includegraphics[width=\linewidth]{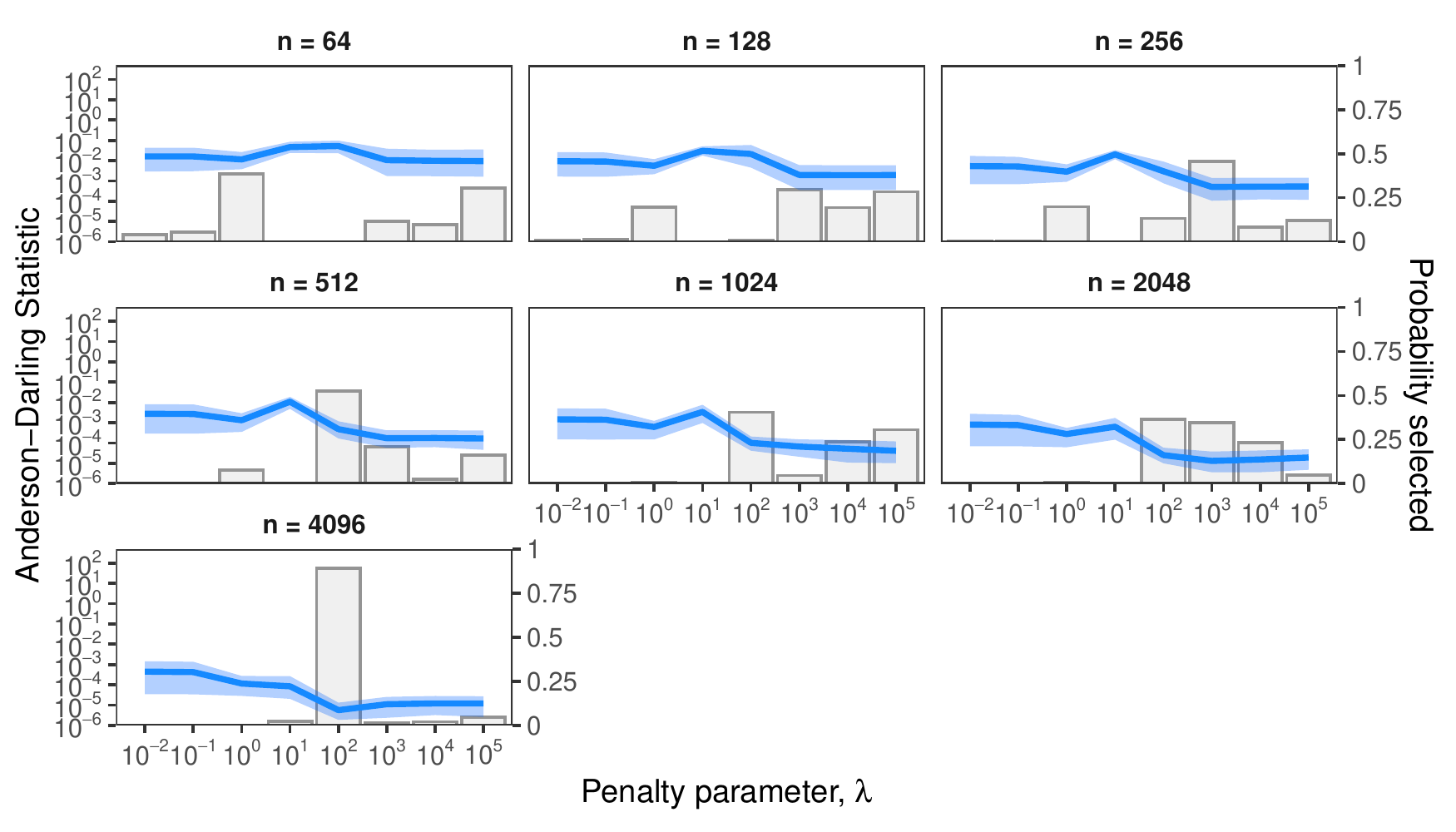}
			\caption{Selection of the penalty parameter, $\lambda$, for the control}
		\end{subfigure}
		\caption{Performance of the tuning algorithm in Algorithm \ref{alg:bootwass}. Discrepancy between the weights estimated by COT and the true self-normalized inverse propensity score in terms of an Anderson Darling statistic in blue. The probability that a $\lambda$ value was selected is given by the histogram.}
		\label{fig:tune}
	\end{figure}

	\section{Additional case study: the  LaLonde Data}\label{sec:lalonde}
	We also validate our method on the LaLonde data set \citep{LaLonde1986}. 
	
	\subsection{The National Supported Work Demonstration program}
	The original data come from a job training program called the National Supported Work Demonstration program (NSW) in which people were randomized to receive or not receive training from the program in the year 1976. The outcome of interest was then to look at the difference in incomes between the treatment and control groups in 1978. The original experimental estimate was a difference of \$1,794 with a confidence interval of $(\$ 551, \$3038 )$. The variables available in the original study include 10 pre-intervention characteristics: earnings and employment in 1974 and 1975, years of education, whether the person received a high school degree, marital status, and indicators for black or Hispanic ethnicity.
	
	\subsection{LaLonde's modification}
	LaLonde then proceeded to modify the original study data by removing the control group and seeing if he could recover the original treatment effect by utilizing an observational data sample taken from the Current Population Survey (CPS) with the same variables measured. This gives 185 participants from the NSM in the treated group and 15,992 non-participants from the CPS in the control group.
	
	\subsection{Methods}
	
	From the Causal Optimal Transport weighting methods, we include no constraints (``none'') and mean constraints (``means''). Hyperparameters were tuned with the algorithm detailed in Algorithm \ref{alg:bootwass}. The distance metric is an $L_2$ metric on the binary covariates and a  Mahalanobis $L_2$ metric on the continuous covariates. We consider the H\'ayek estimator in  \eqref{eq:hajek_init}, a doubly robust/augmented estimator using linear regression with linear terms of the covariates, a weighted least squares estimator, and the barycentric projection estimator of Eq. \eqref{eq:map} utilizing an assignment matrix $\boldP$ constructed utilizing an $L_1$ cost.

	\subsection{Design diagnostics}
	We now display the before and after weighting balance in variable means and 2-Sinkhorn divergence to give a sense of distributional balance. We can see that for all weighting methods both means and distributions are much more similar after weighting than before (Figure \ref{fig:lalonde_diag}).
	
	\begin{figure}[!ht]
		\centering
		\begin{subfigure}[t]{0.6\textwidth}
			\includegraphics[width =\textwidth]{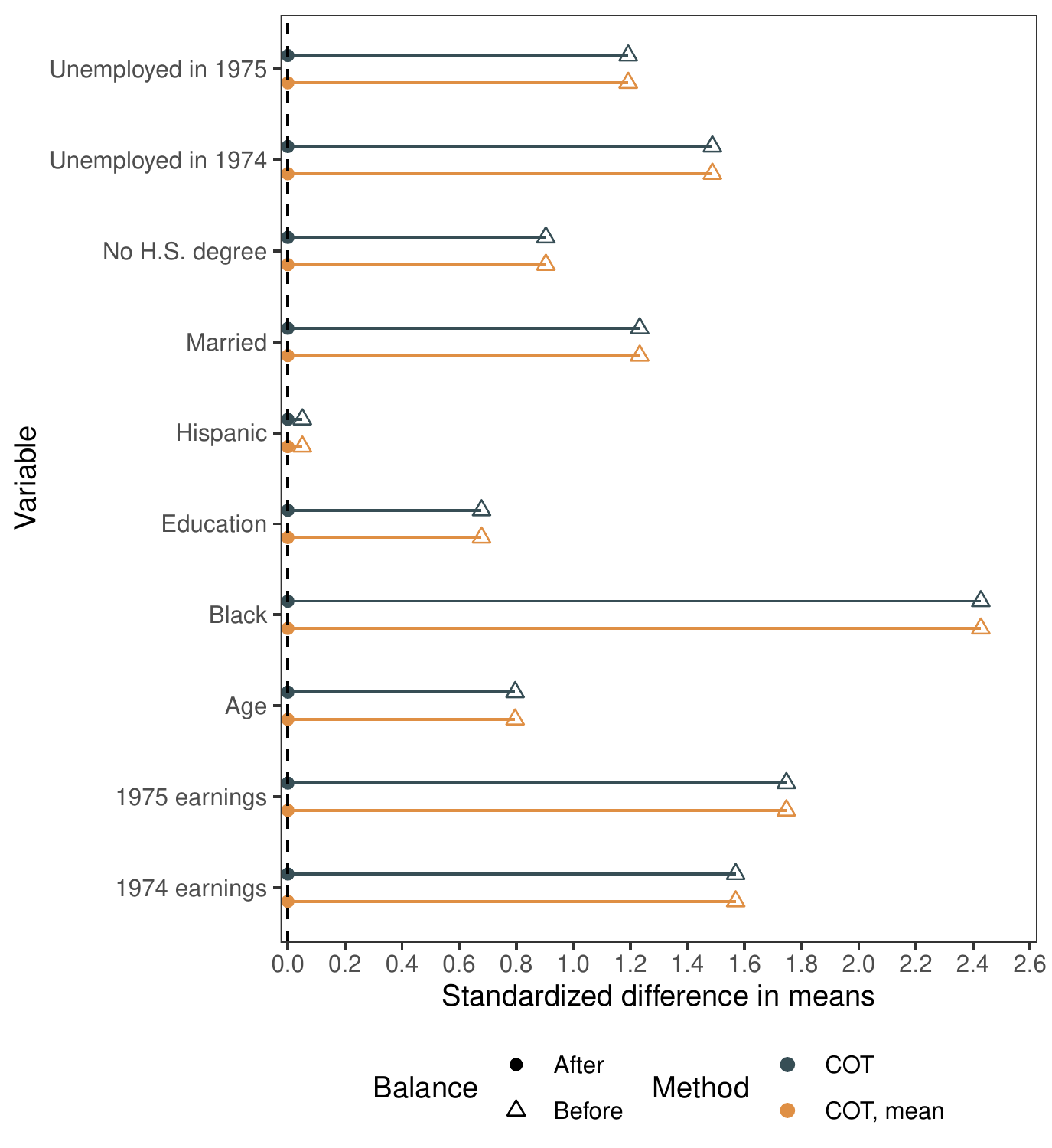}
			\caption{Change in the standardized difference in means between the two groups before and after weighting}
		\end{subfigure}
		\begin{subfigure}[b]{0.7\textwidth}
			\includegraphics[width =\textwidth]{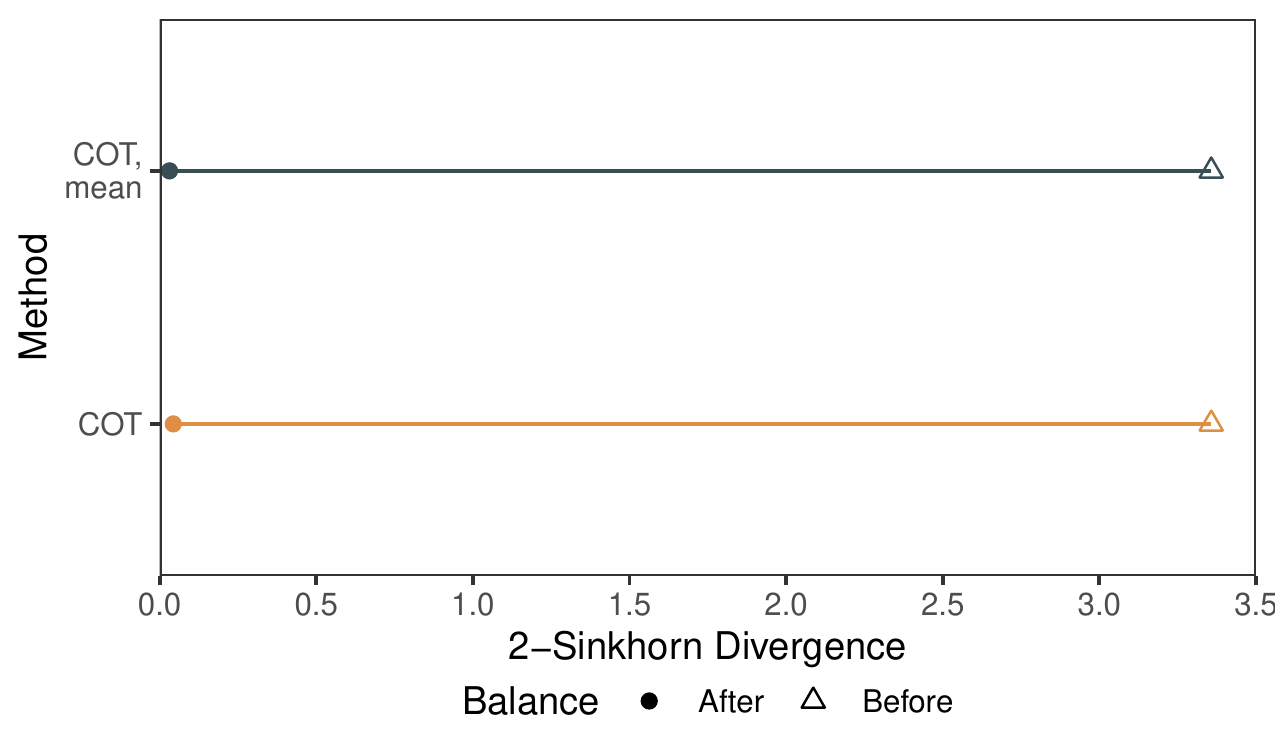}
			\caption{Change in the 2-Sinkhorn divergence between the two groups before and after weighting}
		\end{subfigure}
		\caption{An examination in the change in balance before and after utilizing the optimal transport methods with the listed constraints for the LaLonde data. ``COT'' corresponds to no constraints and ``COT, mean'' corresponds to constraints on the mean balance between distributions.}
		\label{fig:lalonde_diag}
	\end{figure}

	\subsection{Results}
	In Table \ref{tab:lalonde} and Figure \ref{fig:lalonde}, we see that we are able to get very close to the original effects for the H\'ajek, Augmented, and weighted least squares approaches. The barycentric projection estimators have a notable upward bias but still have confidence intervals covering the true effect.

	\begin{figure}[!htb]
		\centering
		\includegraphics[width =0.7\textwidth]{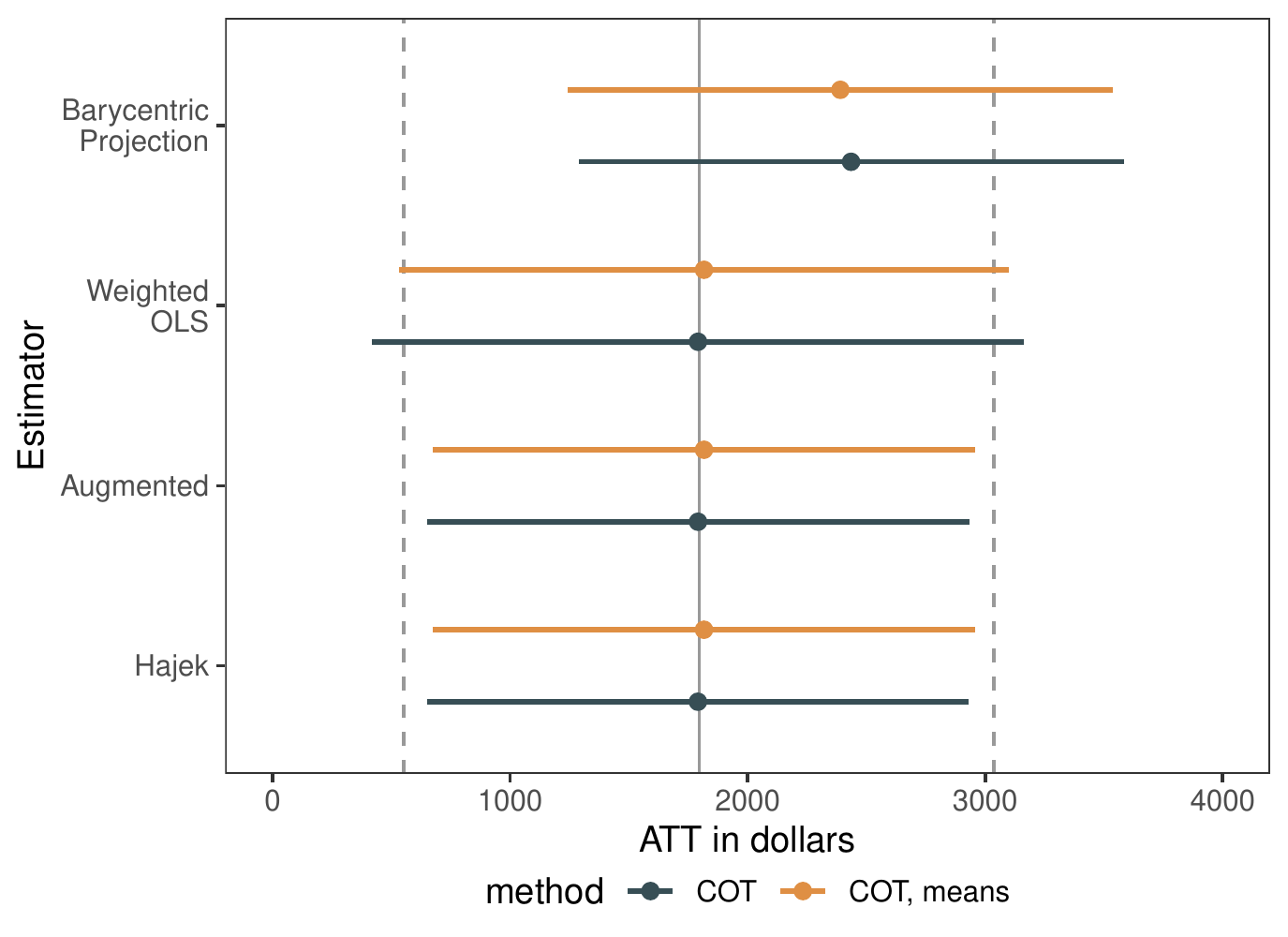}
		\caption{Results for treatment effect estimation for the National Work Support demonstration treated group and the weighted set of controls from the Current Population Survey. The estimate is the difference in 1978 earnings in dollars between the two groups targeting the average treatment effect of the treated (ATT). We see that all optimal transport methods and estimators displayed are able to get close to the original treatment effect. Note that ``COT'' corresponds to no constraints and ``COT, mean'' corresponds to constraints on the mean balance between distributions.}
		
		\label{fig:lalonde}
	\end{figure}

	\input{tables/lalonde_att.tex}

	\section{Other versions of Causal Optimal Transport} \label{sec:other_ot}
	
	We can also represent COT using Eq. \eqref{eq:ot_reg}:
	\begin{equation}
		\cotobj{\bolda}{\lambda} = \min_{\boldw \in \Delta_n} \; \ot{\boldw_1}{\bolda}{\lambda} + \ot{\boldw_0}{\bolda}{\lambda}.
		\label{eq:COT_pen}
	\end{equation}
	In these equations, we can either use an entropy or an $L_2$ penalty and can also incorporate balancing constraints.
	We can, of course, show that Theorems \ref{thm:conv_cot} and \ref{thm:vopt} hold. 
	
	\subsection{Other versions of COT also converge}
	Starting with the proof of convergence, we need the following additional assumption:
	\begin{assumption}
		\label{assum:conv_append}
		For $\cotobj{\bolda}{\lambda}$ in Eq \eqref{eq:COT_pen} with an entropy penalty: $\lambda \to 0$ as $n \to \infty$
	\end{assumption} 
	
	\begin{proof}
		We begin by proving the $L_2$ regularized weights converge, then the entropically regularized weights, and finally, the Sinkhorn divergence. We also have that under Assumption \ref{assum:si}, $\boldw^\star$ exists. Then by Lemma \ref{lemm:is_conv}, $\boldw^\star \rightharpoonup \alpha$.
		
		\textbf{$\mathbf{L_2}$ penalization.}
		Theorem 1 of \citet{Blondel2018} give bounds on $\ot{\bolda_z}{\bolda}{\lambda}$:
		\begin{equation}
			\frac{\lambda}{2} \sum_{i: Z_i = z,j} \left(\frac{\bolda_{z,i}}{n} + \frac{\bolda_j}{n} - \frac{1}{n^2} \right)^2 \leq \ot{\boldw^\star}{\bolda}{\lambda} - \ot{\bolda_z}{\bolda}{} \leq \frac{\lambda}{2} \min \{ \| \bolda_z \|^2, \|\bolda \|^2 \}. 
			\label{eq:l2_bounds}. 
		\end{equation}
		
		\noindent Then the upper bounds on the $L_2$ regularized problem for the importance sampling weights are
		\[ \ot{\bolda_z}{\bolda}{\lambda}  \leq \ot{\boldw^\star}{\bolda}{} + \frac{\lambda}{2} \|\bolda \|, \]
		where the first inequality follows from rearrangement of Eq. \eqref{eq:l2_bounds} and the fact that $\min \{ \| \boldw^\star \|^2, \|\bolda \|^2 \}$ is minimized by the measure where all the observations have the same weight.
		Also, 
		\[\|\bolda \|^2 = \sum_j \bolda_j^2 = \frac{1}{n} \to 0. \]
		Thus, $\ot{\boldw^\star}{\bolda}{\lambda} \to \ot{\boldw^\star}{\bolda}{}$ and by Corollary 6.9 of \citet{villani_cedric_optimal_2008}, $\ot{\boldw^\star}{\bolda}{} \to 0$.
		
		Now we turn directly to the Causal Optimal Transport weights. The problem is convex \citep{Blondel2018}, which means that 
		\[ \ot{\boldw_{\text{COT}}}{\bolda}{\lambda} \leq \ot{\boldc}{\bolda}{\lambda} \] for all $\boldc \in \Delta_n$ that satisfy the constraints of the problem. 
		Further, by assumption $\exists n > 0$ such that the importance sampling weights $\boldw^\star$ also satisfy the balancing constraints. This means that 
		\[ \ot{\boldw_{\text{COT}}}{\bolda}{\lambda} \leq \ot{\boldw^\star}{\bolda}{\lambda} \]
		and both quantities also satisfy the problem constraints for some $n$.
		
		Finally, 
		if $\ot{\boldw_{\text{COT}}}{\bolda}{\lambda}$ goes to 0, this will mean $\boldw_{\text{COT}} \rightharpoonup \alpha$ since $\ot{\boldw_{\text{COT}}}{\bolda}{\lambda} \to \ot{\boldw_{\text{COT}}}{\bolda}{}$.
		Thus, since $\ot{\boldw_{\text{COT}}}{\bolda}{\lambda} \to 0$ because $\ot{\boldw^\star}{\bolda}{\lambda}$ goes to 0, by 
		Corollary 6.9 in \citet{villani_cedric_optimal_2008}
		\[\boldw_{\text{COT}} \rightharpoonup \alpha.\]
		
		\textbf{Entropy penalization.}
		The entropy penalized Causal Optimal Transport problem is also a convex problem, which allows us to conclude
		\[\ot{\boldw_{\text{COT}}}{\bolda}{\lambda} \leq \ot{\boldc}{\bolda}{\lambda}\]
		for $\forall \boldc \in \Delta_n$ since $\boldw_{\text{COT}}$ minimizes this loss.
		This gives us the bound
		\[
		0 \leq \ot{\boldw_{\text{COT}}}{\bolda}{\lambda} \leq \ot{\boldw^\star}{\bolda}{\lambda}.
		\]
		Since the entropy penalized optimal transport problem does not metrize weak convergence, we require that $\lambda \to 0$.
		
		As $\lambda \to 0$ (by assumption) and $n \to \infty$,
		\[\ot{\boldw^\star}{\bolda}{\lambda} \to 0\]
		since $\boldw^\star \rightharpoonup \alpha$. This implies that
		\[\ot{\boldw_{\text{COT}}}{\bolda}{\lambda} \to 0,\] which implies that $\boldw_{\text{COT}} \rightharpoonup \alpha$ by 
		Corollary 6.9 in \citet{villani_cedric_optimal_2008} since at $\lambda = 0$, $\text{OT}_\lambda = \text{OT}$.

	\end{proof}
	
	\subsection{Convergence happens at a $\sqrt{n}$-rate}
	Then semiparametric efficiency also holds with the following additional assumption:
	\begin{assumption}
		For $L_2$ penalized weights, $c(\cdot, \cdot) = d_\mcX(\cdot, \cdot)^p$, with $p > d/2$ and $\E|X|^q < \infty$ for $q > 2p$.
		\label{assum:conv_rootn_l2}
	\end{assumption}
	\begin{proof}
		First, $\boldw^\star$ exist under Assumption \ref{assum:si} and $\boldw^\star \rightharpoonup \alpha$ by Lemma \ref{lemm:is_conv}. Also, under Assumption \ref{assum:conv}, Theorem \ref{thm:conv_cot} holds and $\boldw_{\text{COT}} \rightharpoonup \alpha$.
		
		\textbf{$\mathbf{L_2}$ regularization.} Theorem 1 of \citet{Blondel2018} give bounds on $\ot{\bolda_z}{\bolda}{\lambda}$:
		\[
		\frac{\lambda}{2} \sum_{i: Z_i = z,j} \left(\frac{\bolda_{z,i}}{n} + \frac{\bolda_j}{n} - \frac{1}{n^2} \right)^2 \leq \ot{\bolda_z}{\bolda}{\lambda} - \ot{\bolda_z}{\bolda}{} \leq \frac{\lambda}{2} \min \{ \| \bolda_z \|^2, \|\bolda \|^2 \}. 
		\]
		
		\noindent This implies that regularized problem converges at a linear rate to the unregularized problem because $\min \{ \| \bolda_z \|^2, \|\bolda \|^2 \} = \|\bolda \|^2$ because under an iid assumption $a_j = 1/n, \forall j$. Therefore, 
		\[
		\ot{\bolda_z}{\bolda}{\lambda} \leq \ot{\bolda_z}{\bolda}{} + \frac{\lambda}{2} \|\bolda \|^2.
		\]
		Then 
		\[
		\|\bolda\|_2^2 = \sum_j \bolda_j^2 = \frac{1}{n} \to 0.
		\]
		This also implies that
		\[\lim_{n \to \infty} \ot{\boldw_{\text{COT}}}{\bolda}{\lambda} = \ot{\alpha}{\alpha}{\lambda} = \ot{\alpha}{\alpha}{}\]
		since $\boldw_{\text{COT}} \rightharpoonup \alpha$.
		
		Further, using Theorem 1 in \citet{Fournier2015} we have that under Assumption \ref{assum:conv_rootn_l2},
		\[\E \{ \ot{\boldw_{\text{COT}}}{\bolda}{}\} \leq \E \{ \ot{\boldw_{\text{COT}}}{\alpha}{} + \ot{\bolda}{\alpha}{}\} = \mcO\left(\frac{1}{\sqrt{n}} \right).\]
		Then,
		\[\E \left\{ \ot{\boldw_{\text{COT}}}{\bolda}{} -  \ot{\alpha}{\alpha}{\lambda} \right\} = \E  \left\{ \ot{\boldw_{\text{COT}}}{\bolda}{}\right\} = \mcO\left(\frac{1}{\sqrt{n}} \right).\]
		
		\textbf{Entropy regularization.}
		First, by convexity
		\[0 \leq \ot{\boldw_{\text{COT}}}{\bolda}{\lambda} \leq \ot{\boldw^\star}{\bolda}{\lambda}\]
		and 
		\[0 \leq \ot{\bolda}{\bolda}{\lambda} \leq \ot{\boldw_{\text{COT}}}{\bolda}{\lambda}.\]
		
		Also,
		\[\ot{\bolda}{\bolda}{\lambda} - \ot{\alpha}{\alpha}{\lambda}  \leq \ot{\boldw_{\text{COT}}}{\bolda}{\lambda} - \ot{\alpha}{\alpha}{\lambda} \leq  \ot{\boldw^\star}{\bolda}{\lambda} - \ot{\alpha}{\alpha}{\lambda}.\]
		Then with Assumptions \ref{assum:si}--\ref{assum:conv_rootn}, the conditions of either Theorem 3 of \citet{Genevay2018} or Corollary 1 of \citet{Mena2019} hold. This means that
		\[ \E   \left\{ \ot{\boldw^\star}{\bolda}{\lambda} -  \ot{\alpha}{\alpha}{\lambda} \right \} = \mcO\left(\frac{1}{\sqrt{n}}\right) \]
		and 
		\[ \E  \left\{\ot{\bolda}{\bolda}{\lambda} -  \ot{\alpha}{\alpha}{\lambda}\right\} = \mcO\left(\frac{1}{\sqrt{n}}\right). \]
		Thus, 
		\[ \E  \left\{ \ot{\boldw_{\text{COT}}}{\bolda}{\lambda} - \ot{\alpha}{\alpha}{\lambda}\right\} = \mcO\left(\frac{1}{\sqrt{n}}\right) .\]
		
	\end{proof}
	
	\subsection{Dual formulation}
	Finally, we also have the following dual form for this problem:
	\begin{theorem}
		The dual of each term in Eq. \eqref{eq:COT_pen} is 
		\begin{equation}
			\max_{g,\xi}  \; g \trans \bolda - \sum_k \delta_k  |\xi_k| - \frac{1}{m}\sum_j \xi \trans B(X_j) - \sum_{i,j'} \indicator(Z_i = z) H^\ast_\lambda \left( g_{j'}  - \xi \trans B(X_{j'})  - \boldC_{ij'} \right),
			\label{eq:COT_pen_dual}
		\end{equation}
		where $B(X) = \begin{pmatrix}
			B_1(X) & ... & B_K(X)
		\end{pmatrix}\trans$ and $H_\lambda^\ast$ is the convex conjugate of the penalty function $H_\lambda$.
		\label{thm:dual}
	\end{theorem} 
	\noindent Note that the convex conjugate of $H_\lambda(x) = \lambda x \log (x) $ is $H_\lambda^\ast(y) =\exp\{(y-1)/\lambda\}$ and the convex conjugate of $H_\lambda(x) = \lambda x^2/2$ is $H_\lambda^\ast(y)=y^2/(2\lambda) $.

	\begin{proof}
		We present the proof of the dual form provided in Theorem \ref{thm:dual}. First, some tools from convex analysis \citep{Boyd2004}.
		
		\paragraph{Strong duality.} If strong duality holds then the value of the primal objective at the optimal primal solution is equal to the dual objective at the optimal dual solution.
		
		\paragraph{Slater's conditions.} Slater's conditions are that the objective function is convex and only has equality and inequality constraints .
		
		\paragraph{Slater's theorem.} If Slater's condition's hold, then strong duality holds.
		
		We are now ready to proceed.
		As a reminder, the primal optimization problem is 
		\begin{alignat*}{2}
			\cotobj{\bolda}{\lambda, z}	 = & \argmin_{\mathbf{P} \geq 0} &&  \sum_{i,j} C_{i,j} {P}_{i,j} + \lambda \frac{1}{2}  P_{i,j}^2   \\
			& \text{subject to } && \sum_{i,j} P_{i,j} \indicator(Z_i = z)   = 1\\
			&& & \boldP \trans \ones_{n} = \bolda \\
			&&& \left|\sum_{i,j} B_k(X_i) P_{i,j} -  \frac{1}{n} \sum_{j'} B_k(X_{j'}) \right|  \leq \delta_k, \, \, \forall k \in \{1,...,K\} .
		\end{alignat*}
		
		We first note that we can separate the basis function constraint into the following two inequality constraints
		\begin{align*}
			\sum_{i,j} B_k(X_i) P_{i,j} -  \frac{1}{n} \sum_{j'} B_k(X_{j'})  &< \delta_k, \\ 
			-\sum_{i,j} B_k(X_i) P_{i,j} + \frac{1}{n} \sum_{j'} B_k(X_{j'})  &< \delta_k .
		\end{align*}
		Further we combine the $k$ basis function upper bounds into one vector $\delta$ and similarly denote $B(X)$ as a $n \times k$ matrix of the basis function constraints and $\mean{B}$ as the average of the basis functions in the target population:   $\frac{1}{m} \sum_{j'} B_k(X_{j'})$.
		Then we re-write the primal problem in its Lagrangian form, defining $\langle x, y \rangle = \operatorname{tr}(x \trans y)$.
		\begin{align*}
			\mcL &= \min_{\mathbf{P} \geq 0} \max_{g,\xi_L, \xi_U} \langle \boldC, \boldP \rangle + 
			\frac{\lambda }{2} \langle \boldP, \boldP \rangle - \langle g, \boldP \trans \ones_n - \bolda \rangle +\\
			& \quad \quad \quad \langle \xi_U , B(X)\trans \boldP \ones_n - \delta  - \mean{B} \rangle + \langle \xi_L , -B(X)\trans \boldP \ones_n - \delta  + \mean{B} \rangle. \\
			\intertext{Because the primal problem contains only equality and inequality constraints and the primal objective is a convex function, then strong duality holds,}
			&=   \max_{g,\xi_L, \xi_U} \min_{\mathbf{P} \geq 0} \langle \boldC, \boldP \rangle + 
			\frac{\lambda }{2} \langle \boldP, \boldP \rangle - \langle g, \boldP \trans \ones_n - \bolda \rangle + \\
			& \quad \quad \quad \langle \xi_U , B(X)\trans \boldP \ones_n - \delta  - \mean{B} \rangle + \langle \xi_L , -B(X)\trans \boldP \ones_n - \delta  + \mean{B} \rangle  \\
			&= \max_{g,\xi_L, \xi_U}  g\trans  \bolda  - (\xi_U  + \xi_L)\trans   \delta - (\xi_U  - \xi_L)\trans  \mean{B} + \\  
			& \quad \quad \quad  \min_{\mathbf{P} \geq 0} \langle \boldC, \boldP \rangle +  \frac{\lambda }{2} \langle \boldP, \boldP \rangle  - g\trans  \boldP \trans \ones_n +(\xi_U  - \xi_L) \trans B(X) \trans \boldP \ones_n \\
			&= \max_{g,\xi}  g\trans  \bolda   -   \delta  \|\xi\|_1  -  \xi \trans \mean{B} + \\  
			& \quad \quad \quad  \min_{\mathbf{P} \geq 0} \langle \boldC, \boldP \rangle +  \frac{\lambda }{2} \langle \boldP, \boldP \rangle  - g\trans  \boldP \trans \ones_n  + \xi \trans B(X) \trans \boldP \ones_n
			\intertext{where we have combined the two positive dual variables $\xi_U$ and $\xi_L$ into one unconstrained variable. Then we can rearrange to get}
			&= \max_{g,\xi}  g\trans  \bolda  -   \delta  \|\xi\|_1  -  \xi \trans \mean{B} +  \\  
			& \quad \quad \quad  \min_{\mathbf{P} \geq 0} \langle \boldC - \ones_n g\trans    +  B(X) \xi \ones_n  \trans , \boldP \rangle +  \frac{\lambda }{2} \langle \boldP, \boldP \rangle .
		\end{align*}
		Then taking the derivative with respect to $\boldP$,
		\begin{align*}
			\nabla_\boldP \mcL  &= \boldC - \ones_n  g \trans   + B(X) \xi  \ones_n \trans + \lambda \boldP\\
			\intertext{and finding the critical point of the gradient gives}
			\boldP  &= \frac{1}{\lambda}\left( \ones_n  g \trans - B(X) \xi \ones_n \trans  - \boldC \right)_+,
		\end{align*}
		where the function $(x)_+ = \max(0,x)$ applied element-wise ensures that the weights are constrained to be positive. Plugging this back in, the objective is now
		\[
		\mcL  = \max_{g,\xi}  g\trans  \bolda  -   \delta  \|\xi\|_1  -  \xi \trans \mean{B}   -\frac{1}{2\lambda}\left( \ones_n  g \trans  - B(X) \xi \ones_n \trans  - \boldC \right)_+^2,
		\]
		as desired.
	\end{proof}
	
\end{document}

%% file: tables/hainmueller.tex
\begin{table}[!htb]
\centering
\begingroup\fontsize{9pt}{10pt}\selectfont
\begin{tabular}{lll|rrr|rrr}
  \hline & & & \multicolumn{3}{c}{Bias} & \multicolumn{3}{|c}{RMSE}\\ \hline
overlap & method & constraint & Hajek & DR & WOLS & Hajek & DR & WOLS \\ 
  \hline
high & GLM & none & \textbf{-0.01} & \textbf{-0.01} & -0.02 & 1.18 & 1.14 & 1.14 \\ 
   & CBPS & means &  0.24 & \textbf{-0.01} & -0.02 & 1.12 & 1.11 & 1.09 \\ 
   & SBW & means & \textbf{-0.01} & \textbf{-0.01} & \textbf{-0.01} & 1.00 & 1.00 & 1.00 \\ 
   & SCM & none &  0.36 &  0.27 &  0.28 & 1.63 & 1.57 & 1.55 \\ 
   & NNM & none &  0.43 &  0.32 &  0.28 & 0.69 & 0.65 & 0.56 \\ 
   & COT & none & \textbf{ 0.01} & \textbf{ 0.01} & \textbf{ 0.01} & 0.61 & 0.61 & 0.61 \\ 
   &  & means & \textbf{ 0.01} & \textbf{ 0.01} & \textbf{ 0.01} & \textbf{0.42} & \textbf{0.42} & \textbf{0.42} \\ 
   \hline
medium & GLM & none &  1.12 &  1.10 &  1.04 & 1.72 & 1.69 & 1.70 \\ 
   & CBPS & means &  1.20 &  1.06 &  0.95 & 1.72 & 1.64 & 1.56 \\ 
   & SBW & means &  0.63 &  0.63 &  0.63 & 1.20 & 1.20 & 1.20 \\ 
   & SCM & none &  1.19 &  1.12 &  1.10 & 2.05 & 1.97 & 1.95 \\ 
   & NNM & none &  0.73 &  0.65 &  0.58 & 0.94 & 0.91 & 0.79 \\ 
   & COT & none &  0.23 &  0.23 &  0.23 & 0.74 & 0.74 & 0.74 \\ 
   &  & means & \textbf{-0.03} & \textbf{-0.03} & \textbf{-0.03} & \textbf{0.43} & \textbf{0.43} & \textbf{0.43} \\ 
   \hline
low & GLM & none &  0.19 &  0.06 &  0.02 & 1.72 & 1.49 & 1.51 \\ 
   & CBPS & means &  0.45 &  0.06 &  0.01 & 1.42 & 1.46 & 1.42 \\ 
   & SBW & means &  0.03 &  0.03 &  0.03 & 1.03 & 1.03 & 1.03 \\ 
   & SCM & none &  0.64 &  0.42 &  0.43 & 1.75 & 1.69 & 1.65 \\ 
   & NNM & none &  0.81 &  0.56 &  0.49 & 1.02 & 0.89 & 0.77 \\ 
   & COT & none &  0.05 &  0.05 &  0.05 & 0.85 & 0.85 & 0.85 \\ 
   &  & means & \textbf{ 0.00} & \textbf{ 0.00} & \textbf{ 0.00} & \textbf{0.41} & \textbf{0.41} & \textbf{0.41} \\ 
   \hline
\hline
\end{tabular}
\endgroup
\caption{Performance of various weighting methods under the simulation settings of \cite{Hainmueller2012}. Bold values are the values with the lowest bias or root mean-squared error (RMSE) of the methods under the same conditions. GLM refers to weighting by the inverse of the propensity score as calculated from a logistic regression model, CBPS is the covariate balancing propensity score, SBW is the stable balancing weights, SCM is the synthetic control method, and COT is the optimal transport formulation proposed in this paper. The estimators are Hajek weights (Hajek), doubly-robust augmented IPW (DR), and weighted least squares (WOLS). All weights are normalized to sum to 1. Constraints refer to balancing constraints and are one of ``none'' for no constraints or ``mean'' for mean constraints.} 
\label{tab:hain}
\end{table}

%% file: tables/pph.tex
\begin{table}[ht]
\centering
\begin{tabular}{l|p{1.75in}p{1.75in}}
  \hline
Method & \% C.I. covering original effect & \% of estimates in original C.I. \\ 
  \hline
GLM & 20 & 20 \\ 
  CBPS & 30 & 30 \\ 
  SBW & 30 & 40 \\ 
  SCM & 50 & 50 \\ 
  NNM & 20 & 20 \\ 
  COT & 60 & 60 \\ 
   \hline
\end{tabular}
\caption{For each method, the table displays the percentage of times that the calculated 95\% confidence interval (C.I.) covered the true treatment effect and whether the estimated treatment effect was inside the original C.I. from the study. The weighting methods under examination are logistic regression (GLM), Covariate Balancing Propensity Score (CBPS), Stable Balancing Weights (SBW), Synthetic Control Method (SCM), Nearest Neighbor Matching (NNM), and Causal Optimal Transport (COT).} 
\label{tab:pph}
\end{table}

%% file: tables/lalonde_att.tex
\begin{table}[ht]
\centering
\begin{tabular}{lllll}
  \hline
 & Hajek & Augmented & Weighted OLS & Barycentric Projection \\ 
  \hline
COT & 1791 (649, 2932) & 1791 (650, 2932) & 1791 (418, 3164) & 2435 (1287, 3583) \\ 
  COT, means & 1816 (675, 2957) & 1816 (675, 2957) & 1816 (532, 3100) & 2390 (1243, 3538) \\ 
   \hline
\end{tabular}
\caption{Results for treatment effect estimation for the National Work Support demonstration treated group and the weighted set of controls from the Current Population Survey. The estimate is the difference in 1978 earnings in dollars between the two groups. Values are estimates with asymptotic 95\% confidence intervals.} 
\label{tab:lalonde}
\end{table}